\documentclass[journal]{IEEEtran}
\usepackage{graphicx}
\usepackage{epstopdf}
\usepackage{cite}
\usepackage{algorithm}
\usepackage{algorithmic}
\usepackage{amsmath}
\usepackage{bm,comment,color,amssymb}
\usepackage{stfloats}
\usepackage{makecell}
\usepackage{comment}
\usepackage{xspace,epsfig,graphicx,subfigure,latexsym,fancyhdr,cite}
\usepackage[colorlinks,citecolor=blue,linkcolor=blue,urlcolor=blue]{hyperref}
\newtheorem{theorem}{Theorem}
\newtheorem{proposition}{Proposition}

\newtheorem{remark}{Remark}
\newtheorem{proof}{Proof}
\begin{document}

	\title{  Fluid Antennas Meet Intelligent Surfaces:\\ Security Analysis of   NOMA Systems Under Hardware Impairments}
	
	\author{Tuo Wu, Jianchao Zheng, Xiazhi Lai, Maged Elkashlan,  \\Hyundong Shin, \emph{Fellow}, \emph{IEEE}, and  Naofal Al-Dhahir, \emph{Fellow, IEEE}
		
	 \thanks{(\textit{Corresponding author: Jianchao Zheng.})}
	 \thanks{ T. Wu is with the School of Electronic and Information Engineering, South China University of Technology, Guangzhou 510640, China   (E-mail: $\rm wtopp0415@163.com$). J. Zheng is with the School of Computer Science and Engineering, Huizhou University, Huizhou 516000, China (E-mail: $\rm zhengjch@hzu.edu.cn$).  X. Lai is with the School of Computer Science, Guangdong University of Education, Guangzhou, Guangdong, China (E-mail: $\rm xzlai@outlook.com$). M. Elkashlan is with the School of Electronic Engineering and Computer Science at Queen Mary University of London, London E1 4NS, U.K. (E-mail: $\rm maged.elkashlan@qmul.ac.uk$). H. Shin is with the Department of Electronic Engineering, Kyung Hee University, Yongin-si, Gyeonggi-do 17104, Korea (E-mail: $\rm hshin@khu.ac.kr$). Naofal Al-Dhahir is with the Department of Electrical and Computer Engineering, The University of Texas at Dallas, Richardson, TX 75080 USA (E-mail: $\rm aldhahir@utdallas.edu$).   }   
		
	}
	
	\markboth{}
{Wu \MakeLowercase{\textit{et al.}}:  Fluid Antennas Meet Intelligent Surfaces: Security Analysis of NOMA Systems Under Hardware Impairments}
	
	\maketitle

	\begin{abstract} 
The revolutionary convergence of fluid antenna systems (FAS) and reconfigurable intelligent surfaces (RIS) creates unprecedented opportunities for secure wireless communications, yet the practical implications of hardware impairments on this promising combination remain largely unexplored. This paper investigates the security performance of non-orthogonal multiple access (NOMA) systems when fluid antennas (FAs) meet intelligent surfaces under realistic hardware constraints. We develop a comprehensive analytical framework that captures the complex interplay between adaptive spatial diversity, intelligent signal reflection, and hardware-induced distortions in short-packet communications. Through novel piecewise linear approximations and block-correlation models, we derive tractable expressions for average secure block error rate (BLER) that reveal fundamental performance limits imposed by hardware impairments. Our analysis demonstrates that while the synergy between FAs and intelligent surfaces offers remarkable degrees of freedom for security enhancement, practical hardware imperfections create performance ceilings that persist regardless of spatial diversity gains. The theoretical framework exposes critical design trade-offs between system complexity and achievable security performance, showing that hardware quality becomes a decisive factor in realizing the full potential of FAS-RIS architectures. Extensive simulations validate our analytical insights and provide practical design guidelines for implementing secure NOMA systems that effectively balance the benefits of fluid-intelligent cooperation against the constraints of realistic hardware limitations.
\end{abstract}
	\IEEEpeerreviewmaketitle

\section{Introduction}

\IEEEPARstart{T}{he} evolution of wireless communications has reached a pivotal moment where two transformative technologies are reshaping secure communications: fluid antenna systems (FAS) and reconfigurable intelligent surfaces (RIS). When fluid antennas (FAs) meet intelligent surfaces, they create unprecedented synergy for secure wireless networks \cite{FAS20,WongK22,NewW24,Zhu241,HHong26}. FAS technology offers dynamic spatial adaptability through multiple switchable ports, delivering exceptional spatial diversity and adaptive beamforming capabilities \cite{Zhu242,TWu3,GhadiF,WangC24,XLai23,LaiX242,HXu24,HXu25,JZheng24,JZheng241,JZheng242,YaoJ24,YaoJ25,NewW23,TWu4,TWu7,LZhou24,SYang25,HHong25,ZZhang26,ZZhang25,YZhu25,THan25}. Complementing this fluid adaptability, RIS transforms the wireless environment into an intelligent, programmable resource through precise control of reflection coefficients, enabling signal enhancement, interference mitigation, and coverage extension \cite{TWu1,TWu2,TWu9,YaoJ241,YaoJ242,YaoJ243,HChen26}.

The convergence of FAS and RIS finds compelling applications in secure non-orthogonal multiple access (NOMA) systems, where their combined adaptability and intelligence can enhance physical layer security by creating favorable conditions for legitimate users while degrading eavesdropper capabilities \cite{DingZ16,DingZ14}. However, the practical deployment of such systems faces significant challenges from hardware impairments that can substantially impact achievable security performance, creating a critical need to understand how these imperfections affect the synergistic benefits when FAs meet intelligent surfaces.

Recent research has explored various aspects of these emerging technologies. For FAS-RIS networks, key contributions include \emph{block-correlation modeling} for tractable analysis \cite{LaiX242,LaiX26,TWu6}, \emph{secrecy capacity derivations} for security applications \cite{JZheng241}, \emph{optimization frameworks} for active RIS scenarios \cite{JZheng242,YaoJ243}, \emph{performance bounds} and algorithms \cite{YaoJ241,YaoJ242,TWu5}, and \emph{spatial correlation characterization} using central limit theorem approximations \cite{GhadiF24}. The integration of FAS with NOMA has focused on enhancing spectral efficiency, with notable advances including \emph{BLER analysis} with exact and asymptotic expressions \cite{JZheng24,TWu4}, \emph{energy harvesting} investigations for wireless-powered systems \cite{GhadiF241}, and \emph{hardware impairment analysis} for multi-user FAS systems \cite{YaoJ25}. Similarly, RIS-assisted NOMA systems have been extensively investigated, with research contributions encompassing \emph{NOMA vs. OMA comparisons} \cite{ZhengB20}, \emph{multi-RIS optimization} strategies \cite{ChengY21,ChengY212}, \emph{short-packet BLER derivations} \cite{VuT22}, and \emph{STAR-RIS characterization} for URLLC applications \cite{XuJ23}. Physical layer security research has increasingly focused on short-packet communications and hardware impairments (HIs) modeling \cite{ZhouG21,ChuZ22,XiaC24}, primarily modeling hardware imperfections as \emph{additive distortion noise} scaling with signal power.

Despite these advances, the comprehensive analysis of FAS-RIS NOMA systems considering both physical layer security and  HIs  remains largely unexplored. This research gap motivates our investigation into how hardware imperfections affect the synergistic benefits when FAs meet intelligent surfaces in secure communication scenarios.

The complexity of analyzing such systems under HIs stems from several interconnected factors. First, the spatial correlation characteristics of FA ports become significantly more complex when coupled with hardware-induced distortions, requiring novel analytical approaches to capture their joint effects on system performance. Second, the intelligent signal manipulation capabilities of RIS are inherently limited by practical hardware imperfections such as phase noise, amplitude errors, and quantization effects, which can substantially degrade the intended signal enhancements. Third, the sophisticated power allocation and successive interference cancellation processes in NOMA systems become particularly vulnerable to HIs, as these imperfections can cascade through the decoding process and compromise security guarantees.

The fundamental question that emerges is whether the remarkable degrees of freedom offered by the FAS-RIS combination can effectively compensate for the performance degradation introduced by HIs, and under what conditions this compensation is most effective. This question is particularly critical for secure communications, where even small performance degradations can have significant implications for system security and reliability. To address these challenges, this paper develops a comprehensive analytical framework that captures the complex interactions between FA adaptability, intelligent surface control, and hardware-induced limitations in secure NOMA systems. The key contributions of this paper are summarized as follows:
	\begin{itemize}
		\item We derive novel closed-form expressions for the average secure BLER of the central user (CU) by employing piecewise linear approximations for BLER functions and leveraging block-correlation approximation (BCA) models for FAS spatial correlation. Our analysis addresses the computational challenges posed by double integrals and correlated decoding errors in successive interference cancellation (SIC) processes, providing tractable expressions that capture the essential system behavior while maintaining analytical accuracy.
		
		\item We provide the first comprehensive analysis of HIs effects on secure FAS-RIS NOMA systems, quantifying how radio frequency (RF) front-end imperfections degrade both reliability and security performance. Our analysis reveals critical design trade-offs and establishes fundamental limits on achievable performance under practical hardware constraints, offering essential insights for system optimization and deployment.
		
		\item Through extensive theoretical analysis and numerical evaluation, we establish practical design guidelines for secure FAS-RIS NOMA systems, including optimal power allocation strategies, FAS port selection criteria, and hardware quality requirements. Our results demonstrate that the proposed system can achieve significant performance improvements over conventional approaches even under severe HIs, providing a roadmap for practical implementations.

	\end{itemize}

\section{System Model}

{\begin{table}[h]
\centering
\caption{Summary of Key Notations}
\begin{tabular}{ll}
\hline
\textbf{Symbol} & \textbf{Description} \\
\hline
$M$ & Number of RIS reflecting elements \\
$L$ & Number of FAS ports \\
$W$ & Normalized FAS antenna size (in $\lambda$) \\
$P$ & Total BS transmit power \\
$a_C, a_E$ & NOMA power allocation coefficients \\
$\varrho_C^2, \varrho_E^2$ & HI distortion levels (CU, EU) \\
$\sigma^2$ & AWGN noise variance \\
$\alpha$ & Path loss exponent \\
$d_{SR_i}, d_{R_i X}$ & BS-to-RIS and RIS-to-user distances \\
$\gamma_O^{(X)}$ & Optimal FAS channel amplitude ($X \in \{C, E\}$) \\
$\gamma_{C,E}, \gamma_{C,C}$ & SINRs for SIC decoding at CU \\
$\gamma_{E,C}, \gamma_{E,E}$ & SINRs at eavesdropper EU \\
$N_c, N_e$ & Message sizes for CU and EU (bits) \\
$m$ & Blocklength (symbols) \\
$\delta$ & Information leakage level \\
$\Psi, \Phi, \Xi$ & BLER components (SIC, secure, interference) \\
$B$ & Number of BCA blocks \\
$U, \tilde{U}$ & Gauss-Chebyshev quadrature orders \\
\hline
\end{tabular}
\end{table}}

To investigate the security performance when FAs meet intelligent surfaces under HIs, we consider a downlink FAS-RIS assisted NOMA secure communication system. The system comprises a  BS with a single fixed-position antenna, a  CU, and EU. In practical urban environments, direct communication links are often blocked by buildings, trees, and other infrastructure. To address this challenge and realize the synergistic benefits of FAs and intelligent surfaces, the system is aided by two RISs, denoted as RIS-1 and RIS-2, each equipped with \(M\) reflecting elements. We assume a practical scenario where direct links from the BS to users are obstructed, and communication is established exclusively through BS-RIS-user paths. Specifically, RIS-1 assists the BS-to-CU link, while RIS-2 assists the BS-to-EU link. The distances from the BS to RIS-1 and from RIS-1 to the CU are \(d_{SR_1}\) and \(d_{R_1C}\), respectively. Similarly, the distances for the EU's link are \(d_{SR_2}\) and \(d_{R_2E}\).

{The two-RIS architecture is motivated by practical deployment considerations in urban environments. In such scenarios, RIS panels are deployed on building facades at strategic locations to create alternative communication paths. The use of separate RIS for the CU and EU links can represent: (i) \emph{spatially separated reflectors}, where the CU and EU are located in different directions from the BS and separate RIS panels deployed on different buildings naturally serve each user's path independently; (ii) \emph{logical partitioning of a large surface}, where a single large RIS is logically divided into sub-arrays with different phase shift profiles to simultaneously serve multiple user directions, with $M$ elements allocated to each partition; or (iii) a \emph{worst-case security model}, where a dedicated RIS-2 for the eavesdropper represents the scenario in which the EU has optimal RIS assistance---in practice, an eavesdropper may have limited or no dedicated RIS support, resulting in better security performance than our conservative analysis predicts. The direct link blockage assumption is widely adopted in RIS-assisted communication literature \cite{VuT22,XuJ23,ChengY21} and represents the primary deployment scenario where RIS provides the most significant benefits.}

Both the CU and EU are equipped with a single  FA  capable of selecting the best receiving port out of \(L\) available ports. These ports are distributed uniformly over a linear space of size \(W\lambda\), where \(\lambda\) is the carrier wavelength and \(W\) is the normalized antenna size.  We also account for HIs (HI) at both users, modeling their effect as additive distortion noise that scales with the received signal power \cite{ZhouG21,ChuZ22,XiaC24}.

The BS employs NOMA for downlink transmission in a short-packet regime to support ultra-reliable low-latency communication (URLLC) applications. A message of $N_c$ bits for the CU and a message of $N_e$ bits for the EU are encoded into a block of $m$ symbols. The superposed signal transmitted by the BS is given by
\begin{align}\label{aq1}
	x=\sqrt{a_C P}x_C+\sqrt{a_E P}x_E,
\end{align}
where \(x_C\) and \(x_E\) are the unit-power signals for the CU and EU, respectively (i.e., \(\mathbb{E}[|x_C|^2] = 1\) and \(\mathbb{E}[|x_E|^2] = 1\)). The power allocation coefficients satisfy \(a_C + a_E = 1\), with \(a_E > a_C\) to ensure proper NOMA decoding order, and \(P\) is the total transmit power at the BS.

For the CU, the received signal at the $l$-th FAS port, after passing through RIS-1, is expressed as 
\begin{align}\label{aq2}
y^{(C)}_{l}=&( \sum_{m=1}^M h^{(C)}_m v^{(C)}_{m,l} e^{-j\pi\theta^{(C)}_m})\left(d_{SR_1}d_{R_1C}\right)^{-\frac{\alpha}{2}}x+z^{(C)}_{l}+z'_{l},
\end{align}
where \(h^{(C)}_m \sim \mathcal{CN}(0, \epsilon^{(C)}_1)\) is the Rayleigh fading channel from the BS to the \(m\)-th element of RIS-1, and \(v^{(C)}_{m,l} \sim \mathcal{CN}(0, \epsilon^{(C)}_2)\) is the channel from that element to the $l$-th port of the CU's FA. The AWGN at the receiver is \(z'_{l} \sim \mathcal{CN}(0, \sigma^2)\). The term \(z^{(C)}_{l}\) represents the distortion noise due to HIs at the CU, which is modeled as $z^{(C)}_{l}\sim\mathcal{CN}\left(0, \varrho_C^2 P \left(\sum_{m=1}^M|h^{(C)}_m||v^{(C)}_{m,l}|\right)^2 \left(d_{SR_1}d_{R_1C}\right)^{-\alpha}\right)$,
where $\varrho_C^2$ is the HIs level that characterizes the quality of the RF front-end components \cite{ZhouG21,ChuZ22,XiaC24}. 

{\begin{remark}[Channel Model and CSI Assumptions]
We adopt Rayleigh flat fading for the BS-RIS and RIS-user channels, which is the standard model for rich-scattering NLOS environments and has been widely used in RIS-aided communication analysis \cite{VuT22,XuJ23,ZhouG21,ChuZ22,XiaC24}. This assumption is particularly appropriate for our system model where direct BS-user links are blocked (necessitating the RIS-assisted paths), as the absence of a dominant line-of-sight component naturally leads to Rayleigh-distributed channel amplitudes. The proposed analytical framework can be extended to Rician fading by modifying the channel statistics (mean and variance of $|h_m^{(X)}||v_{m,l}^{(X)}|$) while preserving the overall piecewise linear approximation and block-correlation model structure. The assumption of perfect CSI for RIS phase optimization serves as a \emph{worst-case security benchmark}: since perfect phase alignment at RIS-2 maximizes the eavesdropper's received signal quality, any practical CSI imperfection at the eavesdropper would only degrade its decoding capability, thereby \emph{improving} the security performance for the legitimate CU.
\end{remark}}

{In practice, transceiver hardware impairments arise from multiple sources including phase noise in local oscillators, in-phase/quadrature (I/Q) imbalance, high-power amplifier (HPA) nonlinearities, and analog-to-digital converter (ADC) quantization errors. The aggregate effect of these individual impairment sources has been rigorously shown to be well-approximated by an additive Gaussian distortion noise whose variance scales with the instantaneous signal power \cite{Schenk08,Bjornson13}. This model---adopted in our work following \cite{ZhouG21,ChuZ22,XiaC24}---provides a tractable yet statistically accurate characterization that captures the essential performance degradation without requiring detailed knowledge of each individual impairment source. The HI level parameter $\varrho^2$ serves as a comprehensive quality metric that can be experimentally calibrated for specific hardware implementations.}

The RIS phase shifts are optimized to maximize the received signal power by aligning the cascaded channels. Specifically, the optimal phase shift for the $m$-th RIS element is given by
\begin{align}
\theta^{(C)}_m = -\text{arg}(h^{(C)}_m v^{(C)}_{m,l}).
\end{align}
This co-phasing strategy results in a real-valued channel amplitude at the \(l\)-th port, which is given by $\gamma^{(C)}_l= \sum_{m=1}^M|h^{(C)}_m||v^{(C)}_{m,l}|$. Thus, the FA selects the port with the maximum channel gain to achieve optimal performance, yielding the optimal channel amplitude 
\begin{align}
	\gamma^{(C)}_O=\mathop{\text{max}}\limits_{l=1,\cdots, L} \gamma^{(C)}_l.
\end{align}
Under the assumption of independent Rayleigh fading channels, the mean and variance of the product \(|h^{(C)}_m||v^{(C)}_{m,l}|\) are given by $\mathbb{E}[|h^{(C)}_m||v^{(C)}_{m,l}|]=\frac{\pi}{4}\sqrt{\epsilon^{(C)}_1\epsilon^{(C)}_2}$, $\mathbf{Var}[|h^{(C)}_m||v^{(C)}_{m,l}|]=\epsilon^{(C)}_1\epsilon^{(C)}_2\left(1-\frac{\pi^2}{16}\right)$.

{\begin{remark}[RIS Phase Quantization]
In practical RIS implementations, the continuous phase shifts $\theta_m^{(X)}$ are quantized to $b$-bit resolution. With $b$-bit phase quantization, the effective channel gain after co-phasing is scaled by a factor $\kappa_b = \mathrm{sinc}(\pi/2^b)$, reducing the optimal channel amplitude to $\gamma_O^{(X),\mathrm{quant}} = \kappa_b \cdot \gamma_O^{(X)}$. This scaling can be directly incorporated into our framework by replacing $|\gamma_O^{(X)}|^2$ with $\kappa_b^2 |\gamma_O^{(X)}|^2$ in all SINR expressions. For $b \geq 3$ bits, $\kappa_b^2 > 0.95$, indicating that the performance loss due to phase quantization is minimal, which justifies our continuous phase shift assumption as a close approximation to practical implementations with moderate quantization resolution. Beyond phase quantization, practical RIS implementations may also exhibit amplitude-phase coupling, inter-element mutual coupling, and phase noise/drift. These effects can be incorporated by modifying the cascaded channel model: the ideal co-phasing gain would be replaced by $\gamma_l^{(X),\mathrm{imp}} = \sum_{m=1}^M \rho_m |h_m^{(X)}||v_{m,l}^{(X)}| e^{j\Delta\phi_m}$, where $\rho_m \leq 1$ captures amplitude loss and $\Delta\phi_m$ captures residual phase error. The CLT approximation remains applicable to this modified sum (with adjusted mean and variance), preserving the overall framework structure. A detailed treatment of these RIS nonidealities in the FAS-RIS-NOMA context is an important direction for future work.
\end{remark}}

The CU performs successive interference cancellation (SIC) to decode the multiplexed signals. First, it decodes the EU's signal \(x_E\) by treating the CU's signal as interference. The SINR for this step is
\begin{align}\label{aq5}
\gamma_{C,E}=\frac{a_E P\left(d_{SR_1}d_{R_1C}\right)^{-\alpha}\left|\gamma^{(C)}_O\right|^2}{(a_C+\varrho_C^2) P \left(d_{SR_1}d_{R_1C}\right)^{-\alpha}\left|\gamma^{(C)}_O\right|^2+\sigma^2}.
\end{align} 

If \(x_E\) is successfully decoded and removed from the received signal, the SNR to decode \(x_C\) is
\begin{align}\label{aq6}
\gamma_{C,C}=\frac{a_C P\left(d_{SR_1}d_{R_1C}\right)^{-\alpha}\left|\gamma^{(C)}_O\right|^2}{\varrho_C^2 P \left(d_{SR_1}d_{R_1C}\right)^{-\alpha}\left|\gamma^{(C)}_O\right|^2+\sigma^2}.
\end{align} 

However, if the decoding of \(x_E\) fails due to channel impairments or noise, the CU must decode \(x_C\) with interference. In this case, the SINR becomes
\begin{align}\label{aq7}
\gamma_{C,CE}=\frac{a_C P\left(d_{SR_1}d_{R_1C}\right)^{-\alpha}\left|\gamma^{(C)}_O\right|^2}{(a_E+\varrho_C^2) P \left(d_{SR_1}d_{R_1C}\right)^{-\alpha}\left|\gamma^{(C)}_O\right|^2+\sigma^2}.
\end{align}

Similarly, the EU receives its signal exclusively through RIS-2. The received signal at the $l$-th FAS port for the EU is
\begin{align}\label{aq8}
y^{(E)}_{l}=&\left( \sum_{m=1}^M h^{(E)}_m v^{(E)}_{m,l} e^{-j\pi\theta^{(E)}_m}\right)\left(d_{SR_2}d_{R_2E}\right)^{-\frac{\alpha}{2}}x+z^{(E)}_{l}+z'_{l},
\end{align}
where the channel coefficients \(h^{(E)}_m \sim \mathcal{CN}(0, \epsilon^{(E)}_1)\) and \(v^{(E)}_{m,l} \sim \mathcal{CN}(0, \epsilon^{(E)}_2)\), HI noise \(z^{(E)}_{l}\), and AWGN \(z'_{l}\) are defined analogously to those for the CU. After optimizing the RIS phases using the same co-phasing strategy, the channel amplitude at the \(l\)-th port is \(\gamma^{(E)}_l= \sum_{m=1}^M|h^{(E)}_m||v^{(E)}_{m,l}|\), and the FA selects the optimal port with gain \(\gamma^{(E)}_O=\mathop{\text{max}}\limits_{l=1,\cdots, L} \gamma^{(E)}_l\).

The EU treats the CU's signal as interference when decoding its own signal. The SINR at the EU to decode its own signal \(x_E\) is
\begin{align}\label{aq11}
\gamma_{E,E}=\frac{a_E P\left(d_{SR_2}d_{R_2E}\right)^{-\alpha}\left|\gamma^{(E)}_O\right|^2}{(a_C+\varrho_E^2) P \left(d_{SR_2}d_{R_2E}\right)^{-\alpha}\left|\gamma^{(E)}_O\right|^2+\sigma^2}.
\end{align} 

As an eavesdropper, the EU may also attempt to decode the CU's confidential message \(x_C\) to compromise the system security. {Following the \emph{worst-case eavesdropper} convention widely adopted in the physical layer security literature \cite{YangW19,LaiX21}, we assume the EU can perform \emph{ideal} SIC to perfectly remove its own signal $x_E$ before decoding $x_C$, even though hardware impairments are present at the EU's receiver. This assumption represents a deliberate security-conservative design choice: it yields the strongest possible eavesdropping capability, providing a \emph{lower bound} on the achievable security performance. In practice, if the EU's SIC is impaired, the residual interference from undecoded $x_E$ would \emph{degrade} the EU's ability to decode $x_C$, resulting in \emph{better} security for the CU than our analysis predicts. The HI at the EU's receiver is still captured through the $\varrho_E^2$ term in the noise floor, which limits the EU's maximum achievable SNR to $a_C/\varrho_E^2$ regardless of transmit power.} The corresponding SNR would be
\begin{align}\label{aq12}
\gamma_{E,C}=\frac{a_C P\left(d_{SR_2}d_{R_2E}\right)^{-\alpha}\left|\gamma^{(E)}_O\right|^2}{\varrho_E^2 P \left(d_{SR_2}d_{R_2E}\right)^{-\alpha}\left|\gamma^{(E)}_O\right|^2+\sigma^2}.
\end{align}

For the sake of analysis, the 3D Clarke's isotropic scattering model {\cite{Clarke68,WongK22}} is adopted for the FAS within the CU and EU. This model accurately captures the spatial correlation characteristics of FAs in realistic propagation environments{, where uniform scattering from all directions in three-dimensional space yields the sinc-function correlation coefficient}. Specifically, the correlation coefficient between ports $\psi$ and $\omega$ is modeled by 
\begin{align}\label{aq14}
	g^{(X)}(\psi,\omega)=\text{sinc}\left(\frac{2\pi (\psi-\omega) W}{L-1}\right),
\end{align}
where $X\in\{C,E\}$ represents the CU and EU, respectively, and $\mathrm{sinc}(x)=\frac{\sin(x)}{x}$ is the sinc function. Note that $\mathrm{sinc}(x)$ is an even function, which indicates that the correlation coefficient matrix is a symmetric Toeplitz matrix as   follows
\begin{equation}\label{aq15}
	\mathbf{\Sigma}^{(X)}\in\mathbb{R}^{L\times L}=
	\begin{bmatrix}
		g_{1,1} & g_{1,2} & \dots & g_{1,L}\\
		g_{1,2} & g_{1,1} & \dots & g_{1,L-1}\\
		\vdots &  \ddots & \vdots \\
		g_{1,L} & g_{1,L-1} & \dots & g_{1,1}
	\end{bmatrix}.
\end{equation} 

{\begin{remark}[Imperfect SIC Extension]
The proposed framework can accommodate imperfect SIC by introducing a residual interference factor $\zeta \in [0,1]$, where $\zeta = 0$ corresponds to perfect SIC and $\zeta = 1$ to no cancellation. Under imperfect SIC, the SNR for decoding $x_C$ after SIC in \eqref{aq6} is modified to
\begin{align}
\gamma_{C,C}^{\mathrm{iSIC}} = \frac{a_C P (d_{SR_1} d_{R_1C})^{-\alpha} |\gamma_O^{(C)}|^2}{(\zeta a_E + \varrho_C^2) P (d_{SR_1} d_{R_1C})^{-\alpha} |\gamma_O^{(C)}|^2 + \sigma^2}.
\end{align}
This modification only changes the parameter $c$ in the critical point expression from $\varrho_C^2$ to $\zeta a_E + \varrho_C^2$, and all subsequent derivations follow identically. Imperfect SIC effectively increases the ``equivalent HI level,'' further degrading CU performance and tightening the SINR ceiling from $a_C/\varrho_C^2$ to $a_C/(\zeta a_E + \varrho_C^2)$. From a security perspective, imperfect SIC at the CU degrades legitimate link reliability without affecting the eavesdropper's capability, worsening the secure BLER.
\end{remark}}

{\begin{remark}[FAS Port Selection vs.\ Conventional Antenna Selection]
While both FAS port selection and traditional antenna selection aim to maximize the received SNR, they differ fundamentally in three aspects: (i) \textbf{Spatial correlation:} In conventional antenna selection with well-separated antennas ($d \geq \lambda/2$), the channels across antennas are typically assumed independent. In FAS, the $L$ ports are densely distributed within a compact space of size $W\lambda$, resulting in \emph{strong spatial correlation} characterized by the sinc-function correlation matrix in \eqref{aq15}. This correlation fundamentally changes the diversity analysis and requires specialized tools such as the block-correlation approximation (BCA) model. (ii) \textbf{Physical implementation:} FAS uses a single reconfigurable antenna element that can switch its effective position among $L$ ports, requiring only \emph{one RF chain}. Conventional antenna selection requires $L$ physically separate antennas with a switching network, consuming more space and hardware resources. (iii) \textbf{Analytical challenge:} The correlated port gains in FAS make the analysis of $\gamma_O^{(X)} = \max_{l=1,\ldots,L} \gamma_l^{(X)}$ significantly more challenging than the independent case, as the joint distribution of correlated order statistics must be characterized.
\end{remark}}

{Practical FAS port selection can be implemented through pilot-based channel estimation, where the FA sequentially estimates the channel quality at each of the $L$ ports using known pilot symbols. Given that FAS port switching can be performed at electronic speeds ($\sim \mu$s for pixel-based or varactor-based implementations \cite{NewW24,TWu3}), all $L$ ports can be probed within a single coherence interval. For URLLC short-packet communications with $m = 200$ symbols, the packet duration at typical sub-6 GHz rates ($\sim 100~\mu$s) is significantly shorter than the channel coherence time ($\sim 1$--$10$ ms at pedestrian speeds), ensuring that the quasi-static block fading assumption holds and the selected port remains optimal throughout the entire packet transmission.}

\subsubsection{BLER and Secure BLER Formulation}

In this paper, we focus on the average secure BLER of the CU. At the CU, from \cite{Poly10,LaiX21}, the BLER to decode \(x_E\) is
\begin{align}\label{aq16}
\Psi \triangleq Q\left ({\frac {\log _{2}(1+\gamma _{C,E})-N_{e}/m}{\sqrt {V(\gamma _{C,E})/m}}}\right ),
\end{align}
where $Q(y)=\frac{1}{\sqrt{2\pi}}\int_y^{\infty} \exp(-\frac{t^2}{2})dt$ denotes the Gaussian Q function and $V(y)=(\log_2 e)^2\cdot(1-(1+y)^{-2})$ denotes the channel dispersion that characterizes the finite blocklength effects.

Provided that the decoding of \(x_E\) is successful, when $\gamma_{C,C}>\gamma_{E,C}$, the secure BLER to decode \(x_C\) is approximated by \cite{YangW19}
\begin{align} \label{aq17}
\Phi \triangleq Q\left ({\frac {\sqrt {m}}{V^{\frac {1}{2}}(\gamma _{C,C})}\left ({\log _{2}\frac {1+\gamma _{C,C}}{1+\gamma _{E,C}}-\frac {\mu V^{\frac {1}{2}}(\gamma _{E,C})}{\sqrt {m}}-\frac {N_{c}}{m}}\right )}\right ),
\end{align}
where $\mu=Q^{-1}(\delta)$ and $\delta$ refers to the information leakage level that quantifies the security performance. When $\gamma_{C,C}\leq\gamma_{E,C}$, the secure BLER to decode \(x_C\) is one \cite{YangW19}, indicating complete security failure. In this article, we assume $\delta$ is fixed in the range of (0,0.5) as a stringent secrecy constraint. 

Note that in \eqref{aq17}, when $\gamma_{C,C}>\gamma_{E,C}$, the EU is still able to decode some portion of messages intended for the CU with errors, since $\log_2\left(1+\gamma_{E,C}\right)$ is positive. However, as stated in \cite{YangW19}, if the BS adds redundancy to the transmitted signals and uses randomness to derive secrecy, reliable secure transmission can be guaranteed. The amount of redundancy depends on $\log_2\left(1+\gamma_{E,C}\right)$, the secure BLER, and the information leakage level, whose relationship is specified by \eqref{aq17} \cite{LaiX21}. 

Because of error propagation in the SIC process, the secure BLER to decode \(x_C\) at the CU is $\epsilon=1-\bar{\epsilon}$ where
\begin{align}\label{aq18}
\bar{\epsilon}=(1-\Psi)(1-\Phi)+\Psi (1-\Xi).
\end{align}

On the right-hand side of \eqref{aq18}, the first term $(1-\Psi)(1-\Phi)$ represents the successful probability of decoding \(x_C\) when both \(x_E\) and \(x_C\) are successfully decoded, while the second term $\Psi (1-\Xi)$ represents the successful probability when \(x_E\) decoding fails but \(x_C\) can still be decoded with interference. According to \cite{LaiX21}, $\Psi$ and $\Phi$ are not independent because of the common random variable $\left|\gamma^{(C)}_O\right|^2$. Similarly, $\Psi$ and $\Xi$ are also not independent due to the shared channel realization.

\section{Performance Analysis}

\subsection{Average Secure BLER Formulation}

The average secure BLER is obtained by taking the expectation of $\epsilon$ with respect to the optimal channel gains $\left|\gamma^{(C)}_O\right|^2$ and $\left|\gamma^{(E)}_O\right|^2$.  Therefore, we have 
\begin{align} \label{aq20}
\mathbb {E}[\epsilon ]\!=\!1\!- \int _{0}^{\infty }\!\!\!\int _{0}^{\infty }\bar {\epsilon }|_{\!\scriptscriptstyle \begin{array}{l} \left|\gamma^{(E)}_O\right|^2=y, \\ \left|\gamma^{(C)}_O\right|^2=z\end{array}\!\!\!}f_{\left|\gamma^{(C)}_O\right|^2}(y)f_{\left|\gamma^{(E)}_O\right|^2} (z)dydz.
\end{align}

Substituting \eqref{aq18} into \eqref{aq20} and utilizing the linearity of expectation, we obtain
\begin{equation} \label{aq21}
\mathbb {E}[\epsilon ]=1-\mathbb {E}[\bar{\epsilon}]=1-\mathbb {E}[(1-\Psi)(1-\Phi)]-\mathbb {E}[\Psi (1-\Xi)].
\end{equation}

Expanding \eqref{aq21}, we get
\begin{equation} \label{aq21a}
\mathbb {E}[\epsilon ]=\mathbb {E}[\Phi ]-\mathbb {E}[\Psi \Phi ]+\mathbb {E}[\Psi \Xi ].
\end{equation}

\subsection{Linear Approximation Method}

The computation of \eqref{aq21a} is challenging because $\Phi$, $\Psi$, and $\Xi$ involve Gaussian Q functions, making the double integrals intractable. To overcome this difficulty, we employ linear approximations for these BLER expressions, which have been shown to provide accurate results in the finite blocklength regime \cite{MakkiB14,LaiX21}.
  The three BLER expressions $\Psi$, $\Phi$, and $\Xi$ all follow the same approximation structure:
\begin{align} \label{aq22}
\mathcal{F}(\gamma) \approx \hat{\mathcal{F}}(\gamma) \triangleq \! \begin{cases} 1; &\quad \gamma \leq \alpha_{low} \\[2pt] \displaystyle \frac {1}{2}-\kappa(\gamma-\alpha_{th}); &\quad \alpha_{low} < \gamma < \alpha_{up} \\[2mm] 0; &\quad \gamma \geq \alpha_{up}, 
\end{cases}
\end{align}
where $\alpha_{low,up}=\alpha_{th}\mp\frac{1}{2\kappa}$ define the linear transition region. For the BLER $\Psi$ of EU signal decoding at the CU, we set $\gamma=\gamma_{C,E}$, $\alpha_{th}=\beta=2^{N_e/m}-1$, and $\kappa=k=(2\pi(2^{2N_e/m-1})/m)^{-1/2}$. For the secure BLER expressions $\Phi$ and $\Xi$ of CU signal decoding, we use $\gamma=\gamma_{C,C}$ and $\gamma=\gamma_{C,CE}$ respectively, with $\alpha_{th}=\tilde{\beta}$ and $\kappa=\tilde{k}=\sqrt{m}(2\pi \tilde {\beta}(\tilde {\beta}+2))^{-1/2}$, where the security threshold parameter is given by
\begin{align} \label{aq24}
\tilde {\beta }=(1+\gamma_{E,C})\exp \left ({\frac {\mu V^{\frac {1}{2}}(\gamma_{E,C})+N_{c}}{\sqrt {m}}\ln 2}\right )-1,
\end{align}
with $\mu=Q^{-1}(\delta)$ and $\delta$ being the target information leakage level.

\subsection{Computation of Individual Expectations}

\subsubsection{Computation of $\mathbb{E}[\hat{\Phi}]$} 
After applying the linear approximation and performing the inner integral, we obtain
\begin{align} \label{aq26}
\mathbb {E}[\hat {\Phi }]\approx \int_0^\infty f_{\gamma_{E,C}}(y) \tilde {k}\int_{\tilde {v}}^{\tilde {u}} F_{\gamma_{C,C}}(x)dx.
\end{align}
Then, we employ the Gauss-Chebyshev quadrature method, which provides excellent accuracy for smooth integrands. Combined with the first-order Riemann integral approximation $\int_a^b f(\tau)d\tau=\frac{b-a}{2}f(\frac{a+b}{2})$, we obtain
\begin{align} \label{aq27}
\mathbb {E}[\hat {\Phi }]\approx &\frac {G}{2}\sum ^{U}_{p=1}\frac {\pi }{U}\sqrt {1-t_{p}^{2}}f_{\gamma_{E,C}} (\frac {G}{2}(t_{p}+1) )F_{\gamma_{C,C}} (\tilde {\beta} ),
\end{align}
where $G$ is a sufficiently large positive number that ensures the integration domain covers the significant portion of the PDF, $U$ is a complexity-vs-accuracy tradeoff parameter, and $t_{p}=\cos \left ({\frac {(2p-1)\pi }{2U}}\right )$  represents the Chebyshev nodes.

\subsubsection{Computation of $\mathbb {E}[\Psi \Phi]$}

The computation of $\mathbb{E}[\Psi \Phi]$ is more challenging due to the correlation between $\Psi$ and $\Phi$ through the common random variable $|\gamma^{(C)}_O|^2$. We employ a conditional expectation approach to handle this dependency.
To derive $\mathbb {E}[\Psi \Phi]$, we first derive $\mathbb {E}[\Psi \Phi|\gamma_{E,C}]$, which is expressed as follows
\begin{align} \label{aq29}
\mathbb {E}[\Psi \Phi |\gamma _{E,C}]\approx \int _{0}^{\infty }\hat {\Psi }\hat {\Phi } \Big |_{\!\scriptscriptstyle \begin{array}{c} \gamma_{E,C}, \\ |\gamma_{O}^{(C)}|^{2}=y \end{array}}\!\!\! f_{|\gamma_{O}^{C}|^{2}}(y)dy.
\end{align}
In \eqref{aq29}, given $\gamma_{E,C}$, $\hat {\Psi }$ and $\hat {\Phi }$ are piecewise functions whose intervals depend on $\gamma_{C,E}$ and $\gamma_{C,C}$. Since both $\gamma_{C,E}$ and $\gamma_{C,C}$ are functions of $|\gamma_{O}^{C}|^{2}$ from \eqref{aq5} and \eqref{aq6}, we need to determine the critical values of $|\gamma_{O}^{C}|^{2}$ at which the piecewise functions switch between different linear segments.

By inverting the SINR expressions at the threshold points, we obtain four critical values that follow the general form:
\begin{align} \label{aq30}
|\gamma_{O}^{(C)}|^{2}=\phi_{t}\triangleq \begin{cases} \displaystyle \frac {\sigma^2 t}{(b-c t) P\left(d_{SR_1}d_{R_1C}\right)^{-\alpha}}; &\quad t < \displaystyle \frac {b}{c} \\[3mm] \displaystyle \infty ; &\quad t\geq \displaystyle \frac {b}{c}, \end{cases}
\end{align}
where $(t,b,c)$ takes the values $(\tilde{v},a_C,\varrho_C^2)$, $(\tilde{u},a_C,\varrho_C^2)$, $(v,a_E,a_C+\varrho_C^2)$, $(u,a_E,a_C+\varrho_C^2)$, $(\tilde{v},a_C,a_E+\varrho_C^2)$, and $(\tilde{u},a_C,a_E+\varrho_C^2)$ for $\phi_{\tilde{v}}$, $\phi_{\tilde{u}}$, $\phi_{v}$, $\phi_{u}$, $\bar{\phi}_{\tilde{v}}$, and $\bar{\phi}_{\tilde{u}}$, respectively.
Thus, $\hat {\Psi }$ and $\hat {\Phi }$ are related to the values of $\phi _{\tilde {v}}$, $\phi _{\tilde {u}}$, $\phi _{v}$, and $\phi _{u}$. By comparing the relative magnitudes of these four critical values, the integration domain is partitioned into distinct cases, each corresponding to different ordering relationships among $\phi _{\tilde {v}}$, $\phi _{\tilde {u}}$, $\phi _{v}$, and $\phi _{u}$. Due to the complexity of the piecewise linear approximations, six different cases need to be considered based on these orderings (detailed derivations are provided in Appendix~\ref{Ap2}). For each case, the conditional expectation $\mathbb{E}[\Psi \Phi |\gamma_{E,C}]$ is computed using Gauss-Chebyshev quadrature integration over the appropriate sub-intervals.

Finally, the unconditional expectation is obtained as 
\begin{align}\label{aq46}
 \mathbb {E}[\Psi \Phi ]\approx \int _{0}^{\infty }\mathbb {E}[\Psi \Phi |\gamma _{E,C}]\Big |_{|\gamma_{O}^{(E)}|^{2}=y} f_{|\gamma_{O}^{(E)}|^{2}}(y)dy.
 \end{align}
Since $\mathbb {E}[\Psi \Phi |\gamma _{E,C}]$ is expressed as a piecewise function with 6 cases, the integral interval $[0,\infty)$ in \eqref{aq46} must be carefully partitioned. From \eqref{aq24}, $\tilde{\beta}$ depends on $\gamma _{E,C}$ and thus on $|\gamma_{O}^{(E)}|^{2}$. The determination of appropriate integration intervals requires solving for the critical points where the ordering relationships among $\phi _{\tilde {v}}$, $\phi _{\tilde {u}}$, $\phi _{v}$, and $\phi _{u}$ change.

To establish the monotonic relationship essential for this analysis, we present the following key result:

\begin{proposition}
The security threshold parameter $\tilde {\beta }$ is a monotonically increasing function with respect to $|\gamma_{O}^{(E)}|^{2}$.
\end{proposition}

\begin{proof}
See Appendix~\ref{Ap1}.
\end{proof}

With Proposition 1, the critical points of $|\gamma_{O}^{(E)}|^{2}$ can be uniquely determined by solving the equations where different critical values become equal. The detailed mathematical derivations for finding these critical points and determining the integration intervals for each case are provided in Appendix~\ref{Ap2}. The solutions are denoted as $\tau_1, \tau_2, \tau_3, \tau_4$ with $\tau_1 \leq \tau_2$ and $\tau_3 \leq \tau_4$.
Based on the relative ordering of these critical points, six distinct cases arise, each defining specific integration intervals for \eqref{aq46}. The detailed analysis of these cases and their corresponding integration domains are provided in Appendix~\ref{Ap2}.

Substituting the integral intervals and $\mathbb {E}[\Psi \Phi]$ for 6 cases into \eqref{aq46}, the integral is complicated since the integrand includes the exponential integral functions. Therefore, we propose to utilize the Gauss-Chebyshev integral to deal with the integral \eqref{aq46}.

In the $i$th case, $i\in\{1,2,\cdots,6\}$, the integral interval can be divided into multiple intervals, where the $ {A}_{i,j}$ and ${B}_{i,j}$ denote the corresponding lower and upper limits. Denote ${\Upsilon }_{i,j}$ as the function $\mathbb {E}[\Psi \Phi|\gamma_{E,C}]$  with respect to$|\gamma_{O}^{(E)}|^{2}$ in the $j$th interval of $i$th case. Thus, $\mathbb {E}[\Psi \Phi]$ can be computed as
\begin{align}
&\mathbb {E}[\Psi \Phi ]\approx \frac {\pi }{\tilde {U}}\sum _{i}\sum _{j}\frac {B_{i,j}-A_{i,j}}{2} \nonumber\\
&\cdot \sum _{p=1}^{\tilde {U}}\Upsilon _{i,j}(\tilde {y}_{p})f_{|\gamma_{O}^{(E)}|^{2}} \left (\tilde {y}_{p}\right )\sqrt {1-\tilde {t}_{p}^{2}},
\end{align}
where $\tilde {U}$ is a complexity-vs-accuracy tradeoff parameter, and
\begin{align} 
\tilde {t}_{p}=&\cos \left ({\frac {(2p-1)\pi }{2\tilde {U}}}\right ),\\ 
\tilde {y}_{p}=&\frac {B_{i,j}-A_{i,j}}{2}\tilde {t}_{p}+\frac {B_{i,j}+A_{i,j}}{2}.
\end{align}

The accuracy of the Gauss-Chebyshev integral can be arbitrary small with large value of $\tilde {U}$ .

\subsubsection{Derivation of $\mathbb {E}[\Psi \Xi]$}

Similarly, to derive $\mathbb {E}[\Psi \Xi]$ , we first derive $\mathbb {E}[\Psi \Xi|\gamma_{E,C}]$, which is expressed as follows
\begin{align} \label{aq70}
\mathbb {E}[\Psi \Xi|\gamma _{E,C}]\approx \int _{0}^{\infty }\hat {\Psi }\hat {\Xi} \Big |_{\!\scriptscriptstyle \begin{array}{c} \gamma_{E,C}, \\ |\gamma_{O}^{(C)}|^{2}=y \end{array}}\!\!\! f_{|\gamma_{O}^{C}|^{2}}(y)dy.
\end{align}
In \eqref{aq70}, given $\gamma_{E,C}$, $\hat {\Psi }$ and $\hat {\Xi}$ are piecewise functions whose intervals depend on $\gamma_{C,CE}$ and $\gamma_{C,C}$. Since both $\gamma_{C,CE}$ and $\gamma_{C,C}$ are functions of $|\gamma_{O}^{C}|^{2}$ from \eqref{aq5} and \eqref{aq7}, we need to determine the critical values of $|\gamma_{O}^{C}|^{2}$ at which the piecewise functions switch between different linear segments.

Following the general form in \eqref{aq30}, when $\gamma_{C,CE}=\tilde {v}$ and $\gamma_{C,CE}=\tilde {u}$, the corresponding critical values are $\bar{\phi}_{\tilde{v}}$ and $\bar{\phi}_{\tilde{u}}$, where $(t,b,c)$ takes the values $(\tilde{v},a_C,a_E+\varrho_C^2)$ and $(\tilde{u},a_C,a_E+\varrho_C^2)$, respectively.

Thus, $\hat {\Psi }$ and $\hat {\Xi }$ are related to the values of $\bar {\phi }_{\tilde {v}}$, $\bar {\phi }_{\tilde {u}}$, $\phi _{v}$, and $\phi _{u}$. By comparing the relative magnitudes of these four critical values, the integration domain is partitioned into six distinct cases (Cases 7-12), each corresponding to different ordering relationships among $\bar {\phi }_{\tilde {v}}$, $\bar {\phi }_{\tilde {u}}$, $\phi _{v}$, and $\phi _{u}$. The detailed derivations for these cases are provided in Appendix~\ref{Ap3}. For each case, the conditional expectation $\mathbb{E}[\Psi \Xi |\gamma_{E,C}]$ is computed using Gauss-Chebyshev quadrature integration over the appropriate sub-intervals.

Similar to the analysis in Section~III-B2, the critical points of $|\gamma_{O}^{(E)}|^{2}$ for Cases 7-12 are determined by solving the equations where different critical values become equal. The solutions are denoted as $\tau_5, \tau_6, \tau_7, \tau_8$ with appropriate ordering relationships (detailed analysis in Appendix~\ref{Ap3}).

Finally, the unconditional expectation $\mathbb{E}[\Psi \Xi]$ is obtained by integrating over the appropriate intervals determined by these critical points. The detailed derivations for Cases 7-12 and their corresponding integration intervals are provided in Appendix~\ref{Ap3}.  Using Proposition 1, $\mathbb {E}[\Psi \Xi]$ can be computed as 
\begin{align}
&\mathbb {E}[\Psi \Xi ]\approx \frac {\pi }{\bar {U}}\sum _{i}\sum _{j}\frac {\bar {B}_{i,j}-\bar {A}_{i,j}}{2} \nonumber\\
&\cdot \sum _{p=1}^{\tilde {U}}\bar {\Upsilon }_{i,j}(\bar{y}_{p})f_{|\gamma_{O}^{(E)}|^{2}} \left (\bar{y}_{p}\right )\sqrt {1-\bar {t}_{p}^{2}},
\end{align}
where $\bar {U}$ is a complexity-vs-accuracy trade-off parameter, $\bar {A}_{i,j}$ and $\bar {B}_{i,j}$ denote the corresponding lower and upper limits of $j$th interval of $i$th case, $i\in\{7,8,\cdots,12\}$, respectively, $\bar {\Upsilon }_{i,j}$ denotes the function $\mathbb {E}[\Psi \Xi|\gamma_{E,C} ]$ with respect to $|\gamma_{O}^{(E)}|^{2}$ in the$j$th interval of $i$th case, and
\begin{align} 
\bar {t}_{p}=&\cos \left ({\frac {(2p-1)\pi }{2\bar {U}}}\right ),\\ 
\bar{y}_{p}=&\frac {B_{i,j}-A_{i,j}}{2}\bar  {t}_{p}+\frac {B_{i,j}+A_{i,j}}{2}.
\end{align}

\subsection{FAS Statistical Characterization}

To complete the analysis in \eqref{aq21a}, we need the PDFs and CDFs of $|\gamma_{O}^{(C)}|^{2}$, $|\gamma_{O}^{(E)}|^{2}$, $\gamma_{E,C}$, and $\gamma_{C,C}$. However, obtaining exact closed-form expressions is intractable due to the max operation in FAS port selection and the cascaded RIS channels.

We employ two approximation techniques: the block-correlation approximation (BCA) model and the central limit theorem (CLT). The BCA model simplifies the FAS correlation structure, while the CLT enables analytical characterization when the number of RIS elements $M$ is sufficiently large.

\subsubsection{Block-Correlation Approximation Model}

The block-correlation approximation (BCA) model \cite{Espinosa24} approximates the Toeplitz correlation matrix $\mathbf{\Sigma}^{(X)}$ in \eqref{aq15} with a block-diagonal structure:
\begin{equation}
  \hat{ \mathbf{\Sigma}}^{(X)}\in\mathbb{R}^{L\times L}=
	\begin{bmatrix}
	\mathbf{A}_1 &\mathbf{0} &\mathbf{0} & \dots &\mathbf{0}\\
	\mathbf{0} &\mathbf{A}_2 &\mathbf{0}& \dots & \mathbf{0}\\
	 \vdots &  \ddots & \vdots \\
	 \mathbf{0}& \mathbf{0} & \mathbf{0}& \dots & \mathbf{A}_B
	\end{bmatrix},
\end{equation}
where each block $\mathbf{A}_b^{(X)}$ has constant correlation coefficient:
\begin{equation}\label{q17}
 \mathbf{A}_b^{(X)}\in\mathbb{R}^{L_b^{(X)}\times L_b^{(X)}}=
	\begin{bmatrix}
	1 &\mu^{(X)} &\mu^{(X)} & \dots &\mu^{(X)}\\
	\mu^{(X)} &1&\mu^{(X)}& \dots &\mu^{(X)}\\
	 \vdots &  \ddots & \vdots \\
	 \mu^{(X)}& \mu^{(X)} & \mu^{(X)}& \dots & 1
	\end{bmatrix},
\end{equation}
where $\mu^{(X)}$ is close to unity, $L_b^{(X)}$ is the block size with $\sum_{b=1}^B L_b^{(X)}=L$, and $B$ is the number of blocks determined by significant eigenvalues. The block sizes are optimized to minimize $\text{dist}( \mathbf{\Sigma}^{(X)}, \hat{ \mathbf{\Sigma}}^{(X)})$ \cite{Espinosa24,LaiX242}.  

\begin{figure*}[!t]
	\normalsize
	\begin{align}\label{eq:gamma_EC_pdf}
		f_{\gamma_{E,C}}(y)=\begin{cases}
			\displaystyle\frac{a_C \sigma^2}{(\varrho_E^2 y-a_C )^2 P \left(d_{SR_2}d_{R_2E}\right)^{-\alpha}} f_{|\gamma_{O}^{(E)}|^{2}}\left(\frac{\sigma^2 y}{(\varrho_E^2 y-a_C)P \left(d_{SR_2}d_{R_2E}\right)^{-\alpha} }\right); & y \leq \frac {a_C}{\varrho_E^2} \\[3mm]
			0; & y > \frac {a_C}{\varrho_E^2}.
		\end{cases}
	\end{align}
	\hrulefill
	\vspace*{4pt}
\end{figure*}
\subsubsection{CLT Approximation}

With large $M$, each channel gain $\gamma_l^{(X)}$ is the sum of $M$ independent random variables, making the central limit theorem (CLT) applicable \cite{LaiX242,LaiX26,TWu6}. {The CLT approximation accuracy improves monotonically with $M$. To quantify this, we performed Kolmogorov-Smirnov (K-S) goodness-of-fit tests comparing the empirical distribution of $\gamma_l^{(X)}$ (from $2\times 10^5$ Monte Carlo samples) with the Gaussian approximation, obtaining K-S statistics of $0.054$, $0.039$, $0.029$, $0.023$, and $0.019$ for $M = 4, 8, 16, 20, 32$, respectively (see Fig.~\ref{fig_clt}). For our default $M = 20$, the K-S statistic of $0.023$ indicates an excellent fit. This is consistent with the general rule-of-thumb for CLT validity ($M \geq 15$--$20$ for sums of non-Gaussian i.i.d.\ variables) and with prior FAS-RIS analysis \cite{LaiX242,LaiX26,GhadiF24}.} The channel gain vector $\mathbf{\Gamma}^{(X)}=[\gamma_1^{(X)},\cdots,\gamma_L^{(X)}]$ is approximated by a Gaussian vector with correlation matrix:
\begin{align}
	\bar{ \mathbf{\Sigma}}^{(X)}\in\mathbb{R}^{L\times L}
	=\begin{bmatrix}
		\mathbf{D}_1^{(X)} &\mathbf{C} &\mathbf{C} & \dots &\mathbf{C}\\
		\mathbf{C} &\mathbf{D}_2^{(X)} &\mathbf{C}& \dots & \mathbf{C}\\
		\vdots &  \ddots & \vdots \\
		\mathbf{C}&\mathbf{C} & \mathbf{C}& \dots & \mathbf{D}_B^{(X)}
	\end{bmatrix},
\end{align}
where $\mathbf{C}$ has inter-block correlation coefficient $\rho_0=\frac{\pi(4-\pi)}{16-\pi^2}$. The channel gain at the $k$-th port is:
\begin{align}
\bar{\gamma}^{(X)}_k=\sqrt{1-\rho^{(X)}_0}d^{(X)}_k+\sqrt{\rho^{(X)}_0}d^{(X)}_0+E_{\gamma_{(X)}},
\end{align}
where $d^{(X)}_k\sim \mathcal{N}(0, V_{\gamma_{(X)}})$, $E_{\gamma_{(X)}}= \frac{M\pi\sqrt{\epsilon_1^{(X)}\epsilon_2^{(X)}}}{4}$, and $V_{\gamma_{(X)}}= M\epsilon_1^{(X)}\epsilon_2^{(X)}\left(1-\frac{\pi^2}{16}\right)$. 
The cumulative distribution function (CDF) of the optimal channel gain $\gamma_{O}^{(X)}=\max_{l=1,\ldots,L}\gamma_l^{(X)}$ is \cite{LaiX242}:
\begin{align}\label{q21}
&F_{\bar{\gamma}_{(X)}^*}(y_{(X)})\nonumber\\
&\approx  \frac{H\pi}{U_l}\sum_{l=1}^{U_l} \frac{1}{2^B}\left[1+\text{erf}\left(\frac{y_{(X)}-E_{\gamma_{(X)}}-\rho^{(X)}_0\tau}{\sqrt{2V_{\gamma_{(X)}}(1-\rho^{(X)}_0)}}\right)\right]^B\nonumber\\
&\quad\times\sqrt{\frac{1-q_l^2}{2\pi V_{\gamma_{(X)}}}}e^{-\frac{\tau^2}{2V_{\gamma_{(X)}}}},
\end{align}
where $H$ is the integration bound, $U_l$ controls accuracy, $\tau=\frac{H(q_l+1)}{2}$, and $q_l=\cos\left(\frac{(2l-1)\pi}{2U_l}\right)$. The probability density function (PDF) is:
\begin{align}\label{q22}
&f_{\bar{\gamma}_{(X)}^*}(y_{(X)})\nonumber\\
&\approx  \frac{H}{U_lV_{\gamma_{(X)}}}\sum_{l=1}^{U_l} \frac{B}{2^B}\left[1+\text{erf}\left(\frac{y_{(X)}-E_{\gamma_{(X)}}-\rho^{(X)}_0\tau}{\sqrt{2V_{\gamma_{(X)}}(1-\rho^{(X)}_0)}}\right)\right]^{B-1}\nonumber\\
&\quad\times\sqrt{\frac{1-q_l^2}{1-\rho^{(X)}_0}}e^{-\frac{(y_{(X)}-E_{\gamma_{(X)}}-\rho_0\tau)^2}{2V_{\gamma_{(X)}}(1-\rho^{(X)}_0)}-\frac{\tau^2}{2V_{\gamma_{(X)}}}}.
\end{align}

The squared channel gains have distributions:
\begin{align}\label{q23}
F_{|\gamma_{O}^{(C)}|^{2}}(t)&=\Pr(|\gamma_{O}^{(C)}|^{2}<t)=\Pr(\gamma_{O}^{(C)}<\sqrt{t})=F_{\bar{\gamma}_{(C)}^*}(\sqrt{t}),\\
f_{|\gamma_{O}^{(C)}|^{2}}(y)&=\frac{1}{2\sqrt{y}}f_{\bar{\gamma}_{(C)}^*}(\sqrt{y}),\\
f_{|\gamma_{O}^{(E)}|^{2}}(z)&=\frac{1}{2\sqrt{z}}f_{\bar{\gamma}_{(E)}^*}(\sqrt{z}).
\end{align}

\begin{figure}[t]
	\centering
	\includegraphics[width=3in]{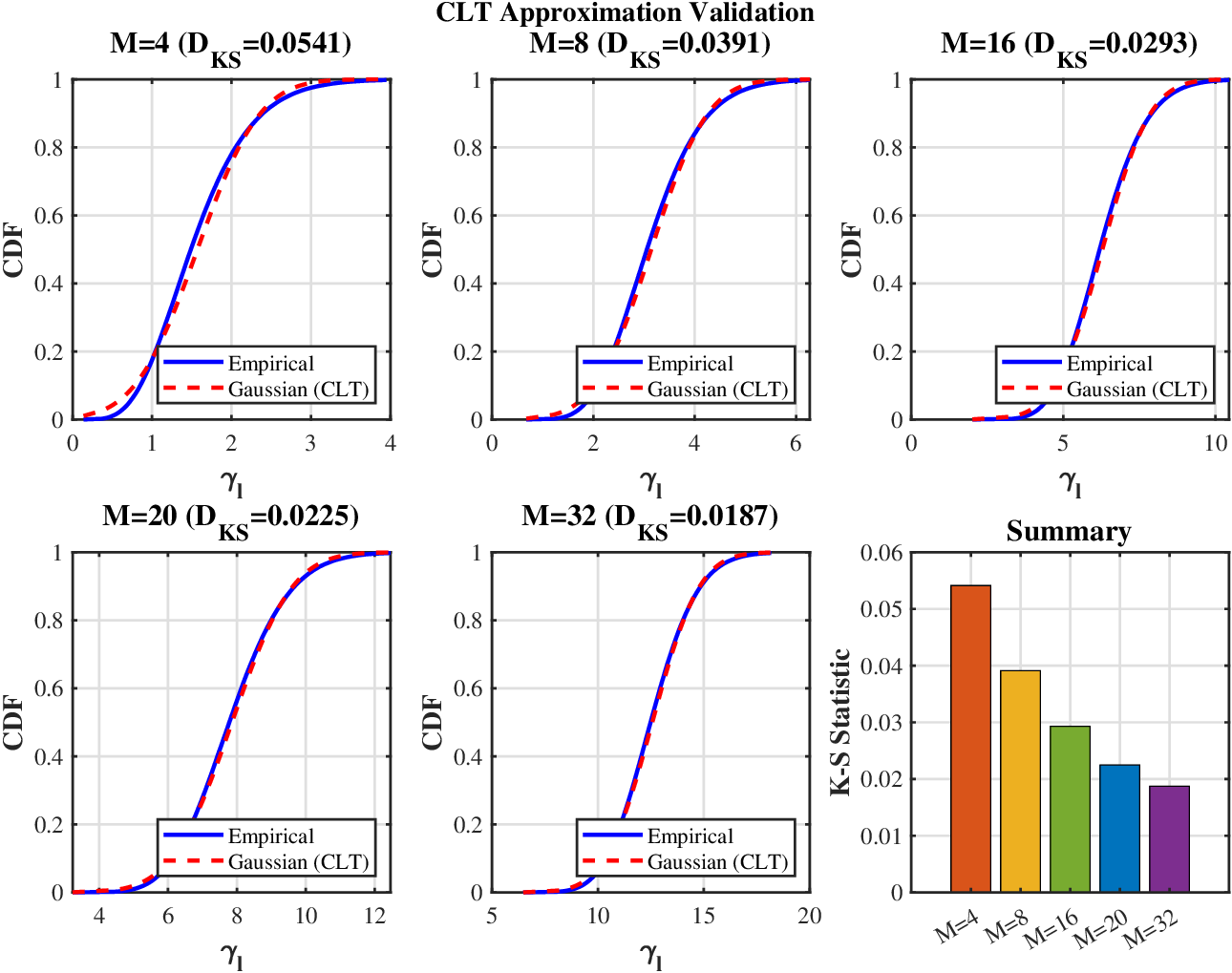}
	\caption{{CLT approximation validation: empirical CDF vs.\ Gaussian CDF of $\gamma_l^{(X)}$ for $M=4, 8, 16, 20, 32$ RIS elements, and K-S statistic summary.}}
	\label{fig_clt}
\end{figure}

The SINR distributions are:
\begin{align}
F_{\gamma_{C,C}}(t)= 
\begin{cases} \displaystyle F_{\bar{\gamma}_{(C)}^*}\left(\sqrt{\frac{\sigma^2\left(d_{SR_1}d_{R_1C}\right)^{\alpha}t}{(a_C-\varrho_C^2 t) P}}\right); &\quad t \leq \displaystyle \frac {a_C}{\varrho_C^2} \\[3mm] \displaystyle 1; &\quad t > \displaystyle \frac {a_C}{\varrho_C^2}. \end{cases}
\end{align}
and the probability density function (PDF) of $\gamma_{E,C}$ is given in (\ref{eq:gamma_EC_pdf}) at the top of this page.

  \subsubsection{Final Average Secure BLER Expression}
 
 Substituting the derived distributions into the expectation expressions from previous subsections yields the final analytical expression:
 \begin{align}
 	\mathbb {E}[\epsilon] = \mathbb {E}[\hat{\Phi}] - \mathbb {E}[\hat{\Psi} \hat{\Phi}] + \mathbb {E}[\hat{\Psi} \hat{\Xi}].
 \end{align} 

\subsection{Computational Complexity Analysis}

While piecewise linearization and Gauss-Chebyshev quadrature make the analysis tractable, the computational complexity remains a practical consideration.   The overall complexity is dominated by the Gauss-Chebyshev quadrature operations, which scale as $\mathcal{O}(U^2)$ for each expectation term, where $U$ is the accuracy parameter. Despite this quadratic scaling, our approach offers significant computational advantages over Monte Carlo simulations, which typically require $\mathcal{O}(10^6)$ samples for comparable accuracy. The proposed analytical framework reduces computation time by approximately two orders of magnitude while maintaining high precision, making it suitable for {offline} system optimization and design parameter exploration.

{More specifically, for each of the three expectation terms $\mathbb{E}[\hat{\Phi}]$, $\mathbb{E}[\hat{\Psi}\hat{\Phi}]$, and $\mathbb{E}[\hat{\Psi}\hat{\Xi}]$, the complexity scales as $\mathcal{O}(U \cdot \tilde{U})$, where $U$ and $\tilde{U}$ are the quadrature accuracy parameters (typically $U = \tilde{U} = 20$--$50$ suffice for high accuracy). The total complexity is therefore $\mathcal{O}(3 \times U \times \tilde{U} \times 6) = \mathcal{O}(18 U \tilde{U})$ for the 6-case piecewise analysis. With $U = \tilde{U} = 30$, this amounts to approximately $1.6 \times 10^4$ function evaluations, which completes in under 0.1 seconds on a standard desktop computer---compared to $\mathcal{O}(10^6)$ Monte Carlo channel realizations that each require FAS port selection, RIS channel generation, and SIC decoding evaluation. To further reduce computational overhead, several approximations can be employed: (i) reducing the quadrature order ($U = 10$--$15$) with modest accuracy loss ($< 3\%$ relative error); (ii) using asymptotic expressions in the high-SNR regime where the SINR ceiling effect simplifies the piecewise analysis; (iii) employing a dominant-case approximation that evaluates only the most probable piecewise case (typically 1--2 out of 6 cases contribute significantly), reducing complexity by $3$--$6\times$.}

{For adaptive real-time system design, a two-tier implementation strategy is proposed: (i) \emph{Offline design}: The analytical expressions are evaluated over a pre-defined grid of system parameters (transmit power $P$, power allocation $a_C$, number of ports $L$, HI levels $\varrho^2$) to construct lookup tables (LUTs) that map system configurations to expected secure BLER performance. (ii) \emph{Online adaptation}: During real-time operation, the system retrieves optimal configurations from the pre-computed LUTs based on current channel statistics, enabling sub-millisecond parameter adaptation. This approach is consistent with practical 5G/6G system design where resource allocation tables are pre-computed and indexed during real-time scheduling.}

\subsection{Impact of HIs on System Performance}

HIs introduce signal-dependent distortion noise that fundamentally alters system performance. The key impacts are: (i) signal-to-interference-plus-noise ratio (SINR) degradation through additional noise terms, (ii) modification of critical points in piecewise linear approximations, and (iii) introduction of performance ceilings.

The SINR degradation ratio compared to the ideal case ($\varrho_C^2 = 0$) is:
\begin{align}
\frac{\gamma_{C,C}^{\text{HI}}}{\gamma_{C,C}^{\text{ideal}}} &= \frac{\sigma^2}{\varrho_C^2 P \left(d_{SR_1}d_{R_1C}\right)^{-\alpha}\left|\gamma^{(C)}_O\right|^2+\sigma^2}, \label{eq:degradation_ratio}
\end{align}
which approaches zero in the high-power regime, indicating a performance ceiling at $\lim_{P \to \infty} \gamma_{C,C}^{\text{HI}} = a_C/\varrho_C^2$.

HI also modifies the critical points in \eqref{aq30} by changing the parameter $c$ from $\{0, a_C, a_E\}$ (ideal case) to $\{\varrho_C^2, a_C + \varrho_C^2, a_E + \varrho_C^2\}$ (HI case), directly affecting the case classifications and integration intervals.

\begin{theorem}[HIs Impact on Average Secure BLER]
\label{thm:hi_impact}
For a FAS-RIS NOMA system with HIs levels $\varrho_C^2$ and $\varrho_E^2$, the average secure BLER satisfies:
\begin{align}
\mathbb{E}[\epsilon^{\text{HI}}] \geq \mathbb{E}[\epsilon^{\text{ideal}}], \quad \text{and} \quad \lim_{P \to \infty} \mathbb{E}[\epsilon^{\text{HI}}] = \mathbb{E}[\epsilon^{\text{ceiling}}] > 0, \label{eq:bler_bounds}
\end{align}
where equality in the first inequality holds if and only if $\varrho_C^2 = \varrho_E^2 = 0$.
\end{theorem}

\begin{proof}
Since  HIs reduce all SINR values and the Q-function is monotonically decreasing, we have $\Psi^{\text{HI}} \geq \Psi^{\text{ideal}}$, $\Phi^{\text{HI}} \geq \Phi^{\text{ideal}}$, and $\Xi^{\text{HI}} \geq \Xi^{\text{ideal}}$. The result follows from \eqref{aq21a} and the SINR ceiling effect.
\end{proof}

\section{Numerical Results} 

To validate our theoretical framework and demonstrate the security performance when FAs meet intelligent surfaces under HIs, we present comprehensive numerical results. The simulations validate the accuracy of our analytically derived average secure BLER expressions and provide insights into the practical implications of HIs on FAS-RIS NOMA systems. For our simulations, we adopt the following parameters: $\epsilon_1^{(C)}={\epsilon_2^{(C)}}=\epsilon_1^{(E)}=1/2$, ${\epsilon_2^{(E)}}=1/20$, $W = 5$, $\alpha=2$, $a_C=0.2$, $a_E=0.8$, $d_{SR_1}=20$ m, $d_{R_1C}=10$ m, $d_{SR_2}=40$ m, $d_{R_2 E}=20$ m, $\varrho^2_C=\varrho^2_E= 0.05$ \cite{ZhouG21}, and path loss factor is 2. The noise variance is $\sigma^2=0.001$ mW, and the BS transmits $N_e=150$ bits to the EU in short packets consisting of $m=200$ symbols.

\begin{figure*}[t]
	\centering
	\begin{minipage}[b]{0.32\linewidth}
		\centering
		\includegraphics[width=2.3in]{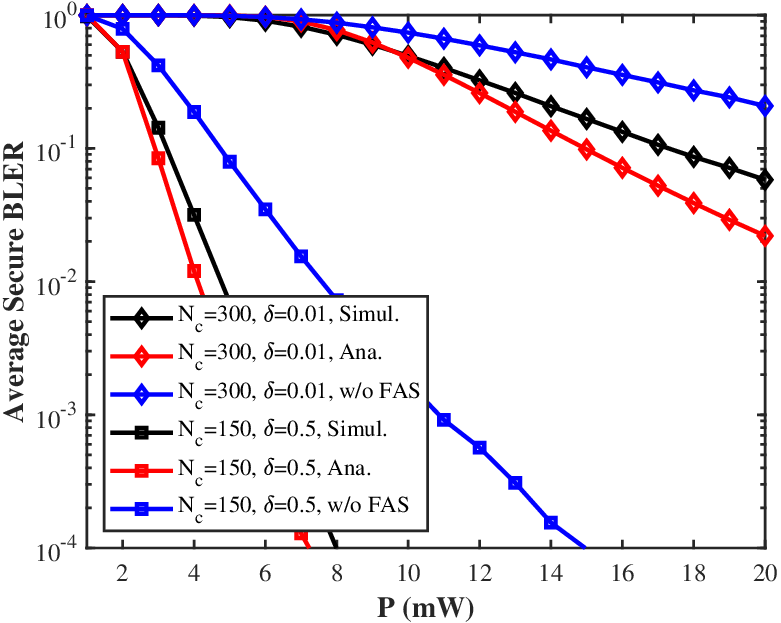}
		\caption{Average secure BLER versus $P$.}
		\label{fig2}
	\end{minipage}
	\hfill
	\begin{minipage}[b]{0.32\linewidth}
		\centering
		\includegraphics[width=2.3in]{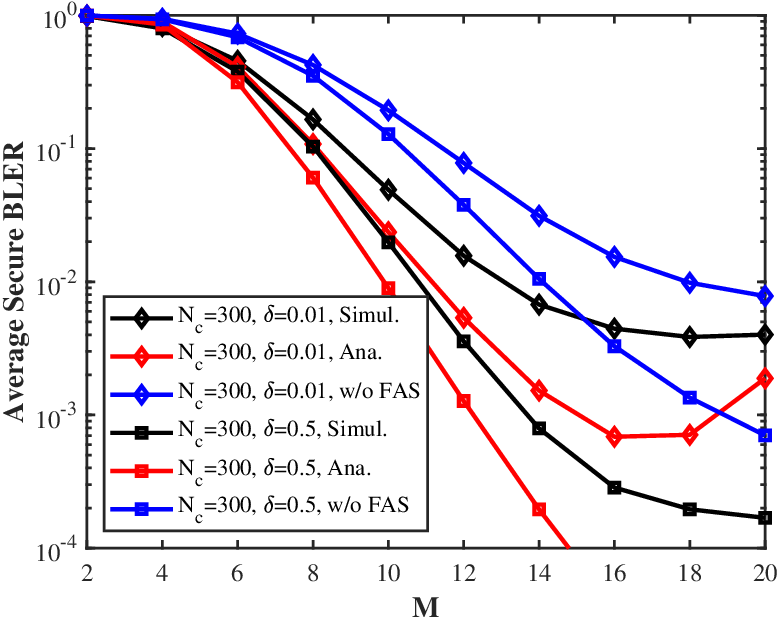}
		\caption{Average secure BLER versus $M$.}
		\label{fig3}
	\end{minipage}
	\hfill
	\begin{minipage}[b]{0.32\linewidth}
		\centering
		\includegraphics[width=2.3in]{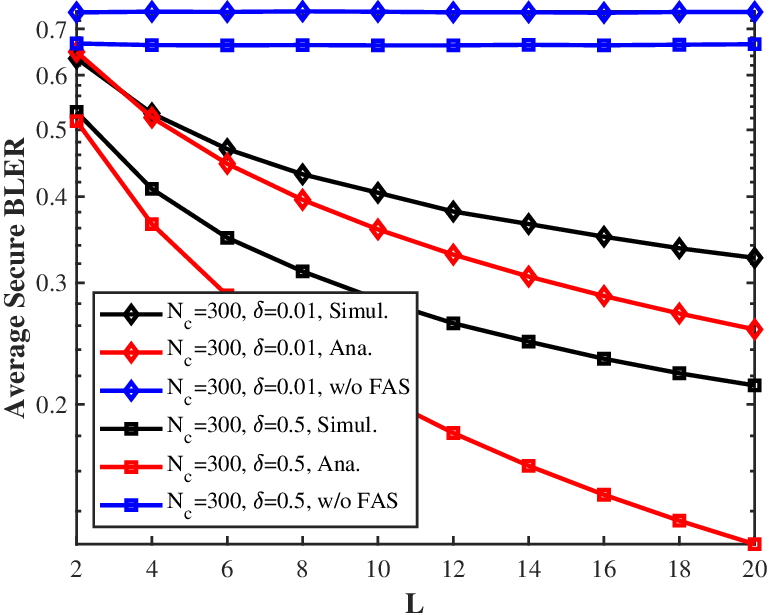}
		\caption{Average secure BLER versus $L$.}
		\label{fig4}
	\end{minipage} \vspace{-4mm}
\end{figure*}

Fig.~\ref{fig2} investigates the average secure BLER performance with respect to transmit power $P$ for multiple parameter configurations under HIs, where $a_C=0.2$, $a_E=0.8$, $\varrho^2_C=\varrho^2_E= 0.05$, $L=5$, and $M=20$. Three distinct scenarios are examined: (i) $N_c=300$ bits with stringent security requirement $\delta=0.01$, (ii) $N_e=150$ bits with relaxed security constraint $\delta=0.5$, and (iii) $N_c=200$ bits with moderate security level $\delta=0.1$.   The results demonstrate that FAS-enabled systems consistently outperform conventional systems across all power levels and parameter settings, with the performance gap being most pronounced at higher transmit powers. 

Fig.~\ref{fig3} investigates the average secure BLER performance with respect to the number of RIS elements $M$ for multiple parameter configurations under HIs, where $a_C=0.2$, $a_E=0.8$, $\varrho^2_C=\varrho^2_E= 0.05$, $L=5$, and $P=100$ mW. Three distinct scenarios are examined: (i) $N_c=200$ bits with stringent security requirement $\delta=0.01$, (ii) $N_c=200$ bits with relaxed security constraint $\delta=0.5$, and (iii) $N_c=300$ bits with relaxed security level $\delta=0.5$.   

Fig.~\ref{fig4} investigates the average secure BLER performance with respect to the number of FAS ports $L$ for multiple parameter configurations under HIs, where $a_C=0.2$, $a_E=0.8$, $P=10$ mW, and $M=20$. Four distinct scenarios are examined: (i) $N_c=300$ bits with stringent security requirement $\delta=0.01$, (ii) $N_c=300$ bits with relaxed security constraint $\delta=0.5$, (iii) $N_c=200$ bits with stringent security $\delta=0.01$, and (iv) $N_c=250$ bits with moderate security level $\delta=0.1$. The results reveal that increasing FAS ports provides diminishing returns, with most performance gains achieved within the first 6-8 ports. 

{From an energy efficiency perspective, both FAS and RIS offer inherent advantages over conventional alternatives: FAS achieves spatial diversity through electronic port switching within a single compact antenna, avoiding the power consumption of multiple independent RF chains required by traditional antenna arrays---for $L$ ports, FAS requires only one RF chain versus $L$ RF chains in conventional selection combining, reducing RF power consumption by a factor of approximately $L$. RIS elements are nearly passive (requiring only low-power PIN diode or varactor biasing, typically $\sim$5--15 mW per element \cite{DiRenzo20}), making them significantly more energy-efficient than active relays or additional base station antennas. The total additional power consumption for a system with $M = 20$ RIS elements is approximately 100--300 mW, which is negligible compared to the BS transmit power. The diminishing returns observed beyond $L = 6$--$8$ FAS ports and the ``cliff effect'' beyond $\varrho^2 = 0.04$ suggest that moderate hardware quality ($\varrho^2 \leq 0.05$) with $L = 6$ ports achieves near-optimal secure BLER performance, offering a practical balance between deployment cost and security performance. Investing in higher RF component quality (lower $\varrho^2$) yields more significant returns than simply adding more FAS ports or RIS elements, providing a clear hardware procurement guideline for system designers.}

\begin{figure*}[t]
	\centering
	\begin{minipage}[b]{0.48\linewidth}
		\centering
		\includegraphics[width=3.2in]{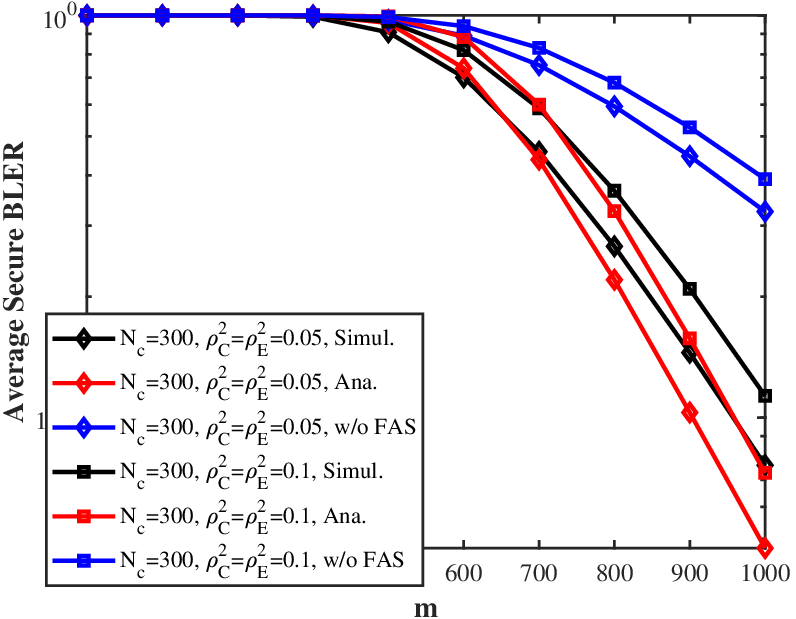}
		\caption{Average secure BLER versus $m$.}
		\label{fig5}
	\end{minipage}
	\hfill
	\begin{minipage}[b]{0.48\linewidth}
		\centering
		\includegraphics[width=3.2in]{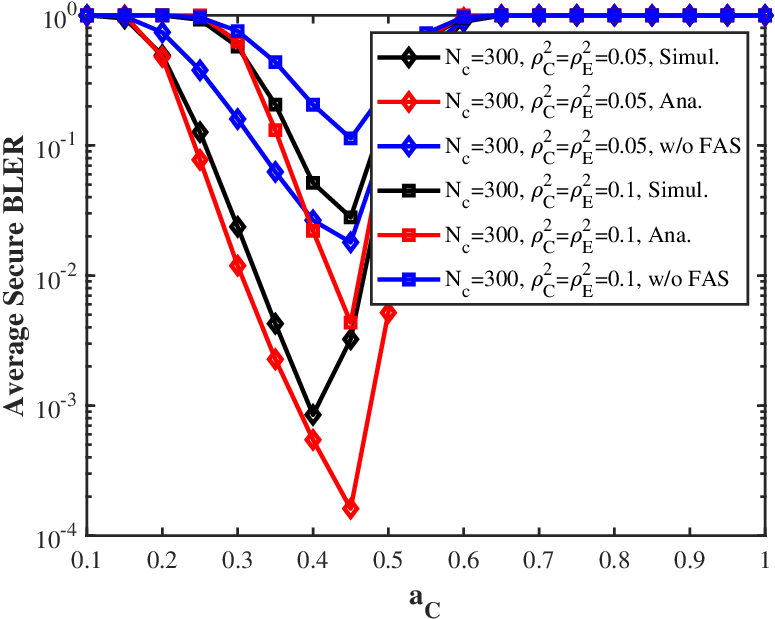}
		\caption{Average secure BLER versus $a_C$.}
		\label{fig6}
	\end{minipage} \vspace{-4mm}
\end{figure*}

Fig.~\ref{fig5} investigates the average secure BLER performance with respect to blocklength $m$ for multiple parameter configurations under HIs, where $a_C=0.2$, $a_E=0.8$, $L=5$, $P=1$ mW, and $M=20$. Four distinct scenarios are examined: (i) $N_c=300$ bits with HIs level $\varrho^2=0.05$, (ii) $N_c=300$ bits with higher impairment $\varrho^2=0.1$, (iii) $N_c=250$ bits with $\varrho^2=0.05$, and (iv) $N_c=200$ bits with $\varrho^2=0.08$. The results demonstrate a clear exponential decay in BLER as blocklength increases, confirming the fundamental finite blocklength theory predictions. Configurations with higher HIs levels ($\varrho^2=0.1$) exhibit consistently worse performance, while larger message sizes ($N_c=300$) require longer blocklengths to achieve the same BLER targets.

Fig.~\ref{fig6} investigates the average secure BLER performance with respect to power allocation coefficient $a_C$ for multiple parameter configurations under HIs, where $a_E=1-a_C$, $P=10$ mW, $L=5$, and $M=20$. Four distinct scenarios are examined: (i) $N_c=300$ bits with HIs level $\varrho^2=0.05$, (ii) $N_c=300$ bits with higher impairment $\varrho^2=0.1$, (iii) $N_c=250$ bits with $\varrho^2=0.05$, and (iv) $N_c=200$ bits with $\varrho^2=0.08$. The results reveal a characteristic U-shaped curve with optimal power allocation occurring around $a_C=0.4-0.5$, representing the fundamental trade-off between signal power and interference in NOMA systems. Configurations with higher HIs consistently exhibit elevated BLER floors, while the optimal operating point shifts slightly toward higher $a_C$ values under severe impairments. The analytical approximations show excellent agreement with simulation results across most operating regions, with minor deviations primarily attributed to the linear approximations $\Phi\approx\hat{\Phi}$, $\Psi\approx\hat{\Psi}$ and the block-correlation approximation (BCA).

  \begin{figure*}[t] 
  	\centering 
  	\includegraphics[width=7.4in]{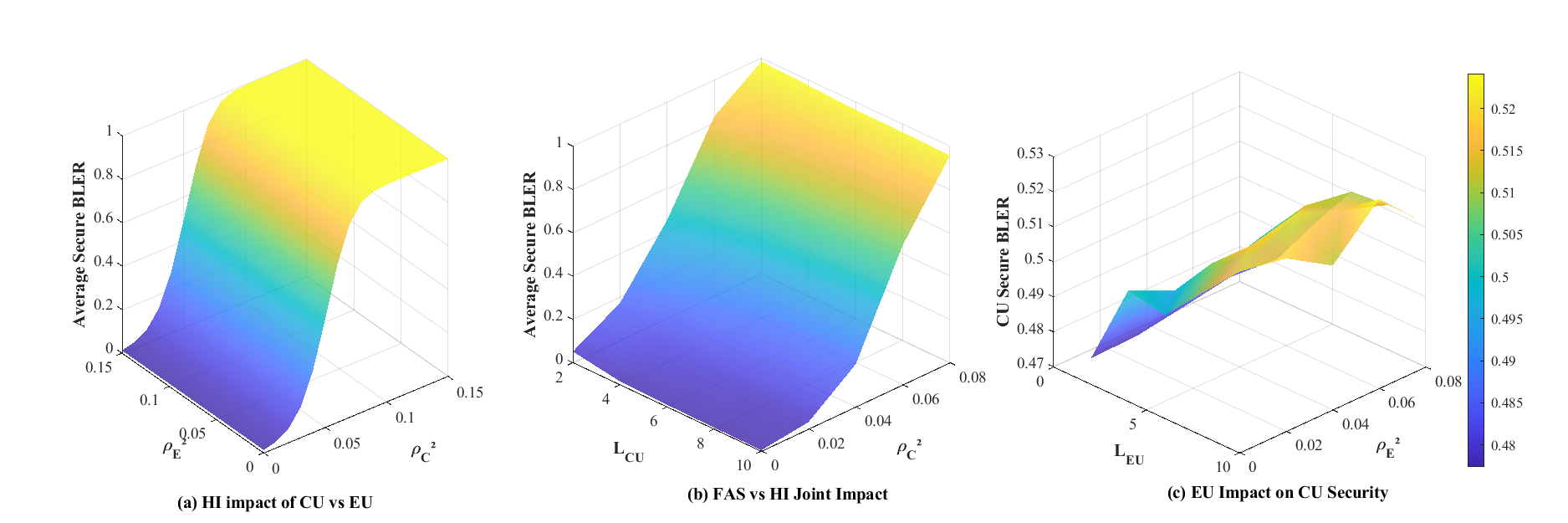}
  	\caption{3D surface analysis of HIs and FAS impact on secure BLER performance.}
  	\label{fig8}
  \end{figure*}

Fig.~\ref{fig8} presents three complementary 3D surface analyses examining HIs and FAS configurations impact on secure BLER performance, where $a_C=0.2$, $a_E=0.8$, $N_c=300$ bits, $N_e=150$ bits, $m=200$, $M=20$, $P=10$ mW, and $\delta=0.01$. Fig.~\ref{fig8}(a) illustrates general 3D BLER performance with FAS ports $L$ (2-10) and HIs levels $\varrho^2$ (0-0.08), showing curved surfaces with exponentially increasing BLER floors at high impairment levels and diminishing FAS returns beyond $L=6$. \emph{HIs create pronounced "cliff" effects beyond $\varrho^2=0.04$, demonstrating critical sensitivity to RF front-end quality. High-quality RF components are essential for practical secure communications.} Fig.~\ref{fig8}(b) depicts CU-specific analysis with CU FAS ports $L_{CU}$ and impairment level $\varrho_C^2$ under fixed EU parameters ($L_{EU}=5$, $\varrho_E^2=0.05$). The surface shows strong coupling between spatial diversity and hardware quality, with optimal regions at high FAS counts and low impairments. \emph{FAS benefits are significantly amplified when HIs are well-controlled ($\varrho_C^2 < 0.03$). Even modest impairments can negate extensive spatial diversity advantages.} Fig.~\ref{fig8}(c) reveals adversarial EU impact on CU security, analyzing EU FAS ports $L_{EU}$ and impairment level $\varrho_E^2$ effects with fixed CU parameters ($L_{CU}=5$, $\varrho_C^2=0.05$). The surface exhibits expected adversarial relationships where improved EU capabilities degrade CU security, with steeper gradients along the HIs axis. \emph{EU hardware quality has more pronounced impact than FAS port count on CU security. Severe EU HIs ($\varrho_E^2 > 0.06$) provide natural security advantages to legitimate links.}

\begin{figure}[t]
	\centering
	\includegraphics[width=3.2in]{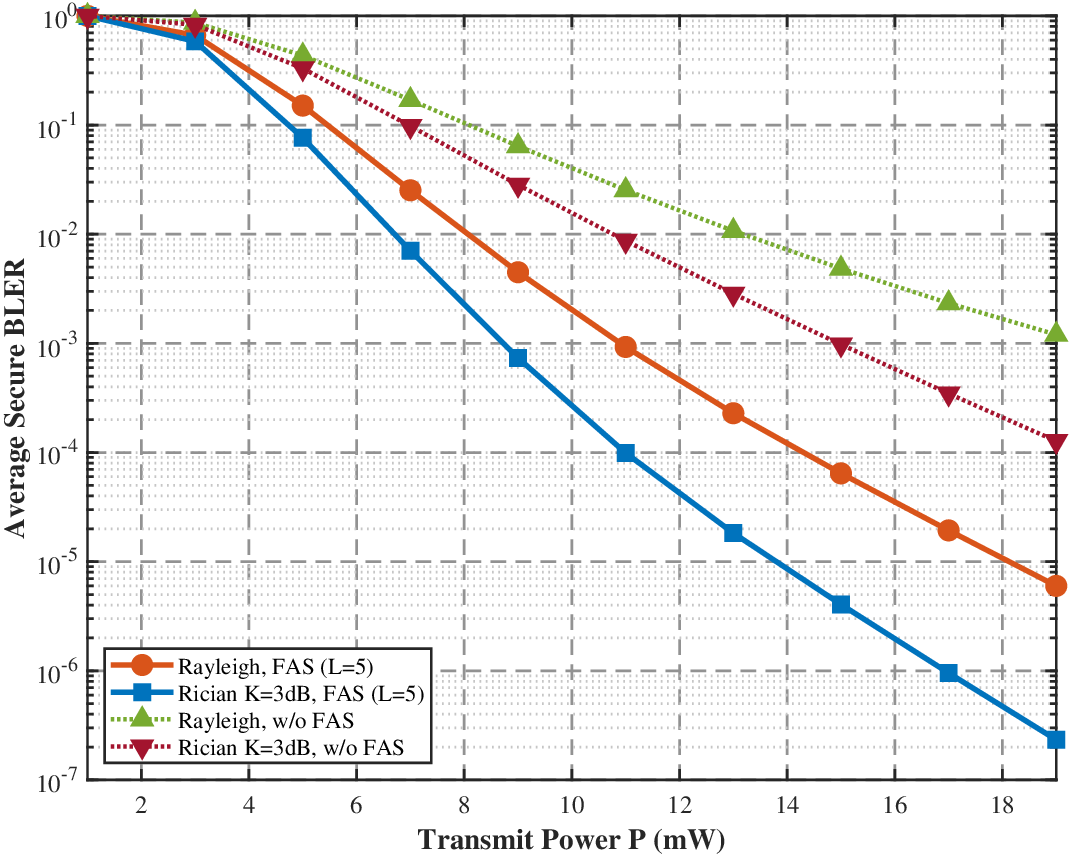}
	\caption{{Average secure BLER versus $P$ under Rayleigh and Rician ($K=3$~dB) fading.}}
	\label{fig_rician}
\end{figure}

{Fig.~\ref{fig_rician} compares the average secure BLER under Rayleigh and Rician ($K=3$~dB) fading for the BS-RIS links, with the RIS-user links remaining Rayleigh (NLOS). The Rician model introduces a line-of-sight (LoS) component in the BS-RIS channels while preserving the same total average channel power. The results demonstrate that: (i) Rician fading yields moderately lower (better) secure BLER than Rayleigh fading across all transmit power levels, as the deterministic LoS component provides a more reliable channel gain; (ii) the FAS diversity advantage persists under both fading models, with FAS-enabled systems consistently outperforming conventional single-antenna systems by approximately one order of magnitude in BLER; (iii) the HI-induced performance ceiling is present under both fading models, confirming that our key design insights---including the diminishing returns of FAS ports and the critical role of hardware quality---remain valid beyond the Rayleigh assumption. These observations validate that the analytical framework and design guidelines developed under Rayleigh fading provide conservative yet representative performance predictions for practical deployments with LoS components.}

\begin{figure*}[t] 
	\centering 
	\includegraphics[width=7.4in]{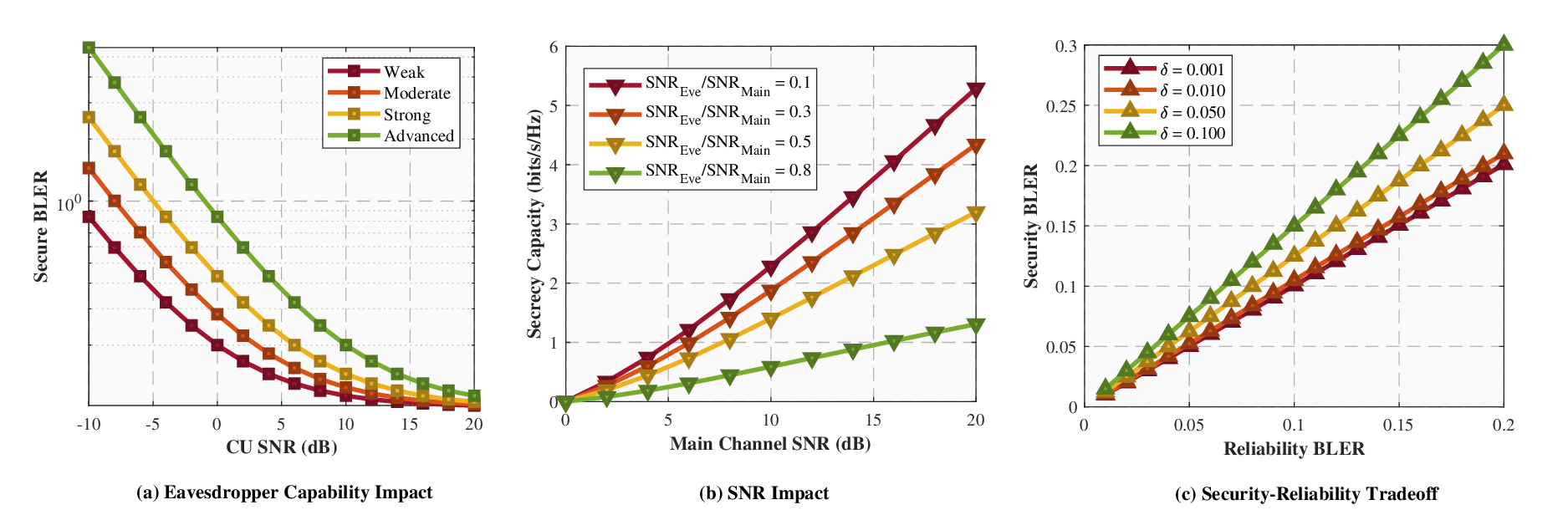}
	\caption{Comprehensive security performance analysis of FAS-RIS NOMA systems.}
	\label{fig9}
\end{figure*}

Fig.~\ref{fig9} presents a comprehensive security performance analysis examining multiple dimensions of secure communication in FAS-RIS NOMA systems. Fig.~\ref{fig9}(a) illustrates eavesdropper capability impact with four threat levels (Weak, Moderate, Strong, Advanced) showing exponential secure BLER degradation as CU SNR increases from -10 to 20 dB. \emph{Advanced eavesdroppers maintain significant threats even at high SNR, while weak adversaries become negligible beyond 10 dB CU SNR. Power control strategies are particularly effective against sophisticated eavesdroppers in low SNR regions.}
Fig.~\ref{fig9}(b) depicts SNR impact analysis with secrecy capacity curves for different eavesdropper-to-main channel SNR ratios (0.1, 0.3, 0.5, 0.8) across 0-20 dB main channel SNR. The curves show logarithmic growth with substantial capacity ($>$4 bits/s/Hz) when eavesdropper channels are weak, but approach zero when channel qualities are comparable. \emph{Channel quality disparity is fundamental to achieving meaningful secrecy rates. Even moderate eavesdropper channel quality severely limits achievable secrecy capacity.}
Fig.~\ref{fig9}(c) presents security-reliability trade-off analysis showing linear relationships between reliability BLER (0-0.2) and security BLER for different information leakage levels $\delta$ (0.001, 0.010, 0.050, 0.100). Stringent security requirements create elevated BLER floors regardless of reliability performance. \emph{Perfect security comes at substantial performance cost, with the trade-off being predictable and optimizable. Moderate security levels offer the best balance for practical applications.}

\section{Conclusion} 

In this paper, we investigated comprehensively the security performance when FAs meet intelligent surfaces in NOMA systems under HIs, revealing how practical RF imperfections fundamentally shape the achievable security performance.   Our comprehensive analysis demonstrates that HIs create performance ceilings that persist regardless of spatial diversity gains, manifesting as exponential BLER floors and diminishing returns beyond optimal FAS port configurations. Despite these limitations, extensive simulations demonstrate that intelligent system design can achieve remarkable security-reliability trade-offs, with key insights including the identification of optimal operating regions, quantification of security-reliability trade-offs under various threat scenarios, and establishment of hardware quality thresholds for practical deployments.  {Additional promising directions include: (i) extending the framework to account for imperfect CSI and channel estimation errors, which would introduce additional performance degradation characterized through mismatched RIS phase configurations and residual interference terms; (ii) incorporating Rician or Nakagami-$m$ fading models for scenarios with dominant line-of-sight components \cite{ZWang25}; (iii) investigating multi-user and multi-cell FAS-RIS NOMA networks, where user scheduling, RIS partitioning, and inter-cell coordination introduce new optimization dimensions \cite{THan25}; (iv) extending the analysis to THz communication bands, which require modifications to account for molecular absorption losses and sparser scattering environments; (v) developing imperfect SIC models as described in Remark~3 and evaluating their impact on security guarantees through the residual interference factor $\zeta$; and (vi) incorporating detailed RIS hardware nonidealities including amplitude-phase coupling and inter-element mutual coupling as discussed in Remark~2; and (vii) exploring joint optimization of communication and computation resources in FAS-RIS aided edge computing systems \cite{TWu8}.}  

\appendices

\section{Critical Points Analysis and Detailed Case Derivations}\label{Ap2}

\subsection{Critical Points Determination}

For the aforementioned 6 cases, the four critical points of $|\gamma_{O}^{(E)}|^{2}$ are the solutions to four equations as follows:
\begin{align} 
\phi _{\tilde {v}}=&\phi _{v},\label{aq47}\\ 
\phi _{\tilde {v}}=&\phi _{u},\label{aq48}\\
\phi _{\tilde {u}}=&\phi _{v},\label{aq49}\\
\phi _{\tilde {u}}=&\phi _{u}.\label{aq50}
\end{align}

Substituting \eqref{aq30} into \eqref{aq47}, we have $\tilde {\beta }-1/(2\tilde{k})=\tilde{\varsigma}_{v}$, where $\tilde{\varsigma}_{v}=\frac{a_C v}{a_E-a_C v}$.
Substituting $\tilde {k}=\sqrt{m}(2\pi \tilde {\beta}(\tilde {\beta}+2))^{-\frac{1}{2}}$ into this equation, after mathematical manipulations, we obtain:
\begin{align} \label{aq52}
(2m-\pi )\tilde {\beta }^{2}-\left ({4m\tilde{\varsigma} _{v}+2\pi }\right )\tilde {\beta }+2m\tilde{\varsigma}_{v}^{2}=0.
\end{align}

From \eqref{aq52}, two solutions of $\tilde {\beta }$ can be obtained. With Proposition 1, after obtaining these solutions, we can determine the solutions of $|\gamma_{O}^{(E)}|^{2}$ to \eqref{aq47}, denoted as $\tau_1$ and $\tau_2$ with $\tau_1\leq \tau_2$, by bisection search. Similarly, the solutions of $|\gamma_{O}^{(E)}|^{2}$ to \eqref{aq49} are also $\tau_1$ and $\tau_2$.

For equations \eqref{aq48} and \eqref{aq50}, we obtain:
\begin{align}\label{aq54}
(2m-\pi )\tilde {\beta }^{2}-\left ({4m\tilde{\varsigma}_{u}+2\pi }\right )\tilde {\beta }+2m\tilde{\varsigma} _{u}^{2}=0,
\end{align}
where $\tilde{\varsigma}_{u}=\frac{a_C u}{a_E-a_C u}$.
Using Proposition 1, the solutions of $|\gamma_{O}^{(E)}|^{2}$ to \eqref{aq54} are denoted as $\tau_3$ and $\tau_4$, with $\tau_3\leq \tau_4$.

\subsection{Integration Interval Analysis}

Based on the relative ordering of the critical points $\tau_1, \tau_2, \tau_3, \tau_4$, six distinct cases arise for the integration intervals in \eqref{aq46}:

{\it{Case 1:}} When $\phi_{\tilde {u}}\leq\phi_{v}$, we have $\tilde {\beta }+1/(2\tilde {k})\leq \tilde{\varsigma}_{v}$, which leads to the inequality $(2m-\pi )\tilde {\beta }^{2}-\left ({4m \tilde{\varsigma} _{v}+2\pi }\right )\tilde {\beta }+2m \tilde{\varsigma}_{v}^{2}\geq 0$ when $0\leq\tilde {\beta }\leq\tilde{\varsigma} _{v}$. The integral interval is $\left ({(0,\tau _{1}]\bigcup [\tau _{2},\infty )}\right )\bigcap [\tau _{0},  {\tau }_{v}]$.

{\it{Case 2:}} The integral interval is $\Theta _{1}\bigcap\Theta _{2}\bigcap\Theta _{3}$, where
\begin{align}
 \Theta _{1}=&\left ({[\tau _{1}, \tau _{2}]\bigcap [\tau _{0}, {\tau }_{v}]}\right )\bigcup [\bar {\tau }_{v}, \infty ),\\
  \Theta _{2}=&\left ({[\tau _{1}, \tau _{2}]\bigcap [{\tau }_{v}, \infty )}\right )\bigcup [\tau _{0}, {\tau }_{v}],\\
   \Theta _{3}=&\left ({(0,\tau _{3}]\bigcup [\tau _{4},\infty )}\right )\bigcap [\tau _{0}, {\tau }_{u}].
 \end{align}
 
{\it{Case 3:}} The integral interval is $\Theta _{1}\bigcap\Theta _{4}$, where $\Theta _{4}=\left ({[\tau _{3}, \tau _{4}]\bigcap [\tau _{0},  {\tau }_{u}]}\right )\bigcup \big[{\tau }_{u}, \infty \big )$.

{\it{Case 4:}} The integral interval is $\Theta _{3}\bigcap\Theta _{5}$, where $\Theta _{5}=\left ({(0,\tau _{1}]\bigcup [\tau _{2},\infty )}\right )\bigcap \big[ {\tau }_{v}, \infty \big)$.

{\it{Case 5:}} The integral interval is $\Theta _{4}\bigcap\Theta _{5}\bigcap\Theta _{6}$, where $\Theta _{6}=\left ({[\tau _{3}, \tau _{4}]\bigcap \big[ {\tau }_{u}, \infty \big )}\right )\bigcup [\tau _{0},  {\tau }_{u}]$.

{\it{Case 6:}} The integral interval is $\left ({(0,\tau _{3}]\bigcup [\tau _{4},\infty )}\right )\bigcap \big[{\tau }_{u}, \infty \big )$.

\subsection{Detailed Case Derivations}

The computation of $\mathbb{E}[\Psi \Phi |\gamma_{E,C}]$ requires careful analysis of different ordering relationships among the critical values $\phi _{\tilde {v}}$, $\phi _{\tilde {u}}$, $\phi _{v}$, and $\phi _{u}$.

{\it{Case 1:}} When $\phi _{\tilde {v}}<\phi_{\tilde {u}}\leq\phi_{v}\leq\phi_{u}$, substituting the linear approximations into the conditional expectation and utilizing mathematical manipulations, we have
\begin{align}
&\mathbb {E}[\Psi \Phi |\gamma _{E,C}]\approx \left(\frac {1}{2}+\tilde {k}\tilde {\beta }\right)\left(F_{|\gamma_{O}^{(C)}|^{2}}\left(\phi _{\tilde {v}}\right)+F_{|\gamma_{O}^{(C)}|^{2}}\left(\phi _{\tilde {u}}\right)\right)\nonumber\\
&+\sum ^{U}_{p=1}\frac {\pi }{U}\frac{\sqrt {1-t_{p}^{2}}}{2}{\frac{a_C P\left(d_{SR_1}d_{R_1C}\right)^{-\alpha}y_p^{(1)}}{\varrho_C^2 P \left(d_{SR_1}d_{R_1C}\right)^{-\alpha}y_p^{(1)}+\sigma^2}}f_{|\gamma_{O}^{(C)}|^{2}}(y_p^{(1)}),
\end{align}
where $y_p^{(1)}=\frac{1}{2\tilde {k}}t_p+\frac{\tilde {\beta }}{2}$.

{\it{Case 2:}} When $\phi _{\tilde {v}}\leq\phi_{\tilde {u}}<\phi_{v}\leq\phi_{u}$,  we have
\begin{align}
&\mathbb {E}[\Psi \Phi |\gamma _{E,C}] \approx   F_{|\gamma_{O}^{(C)}|^{2}}(\phi _{\tilde {v}})+\frac{\phi _{v}-\phi _{\tilde {v}}}{2}\sum ^{U}_{p=1}\frac {\pi }{U}\sqrt {1-t_{p}^{2}} \nonumber\\
&\cdot \left ({\frac {1}{2}-\tilde {k}\left ({\frac{a_C P\left(d_{SR_1}d_{R_1C}\right)^{-\alpha}y_p^{(2)}}{\varrho_C^2 P \left(d_{SR_1}d_{R_1C}\right)^{-\alpha}y_p^{(2)}+\sigma^2}-\tilde {\beta }}\right )}\right )\nonumber\\
 & \cdot  f_{|\gamma_{O}^{(C)}|^{2}}(y_p^{(2)})+\frac{\phi _{\tilde {u}}-\phi _{v}}{2}\sum ^{U}_{p=1}\frac {\pi }{U}\sqrt {1-t_{p}^{2}} \nonumber\\
&\cdot \left ({\frac {1}{2}-\tilde {k}\left ({\frac{a_C P\left(d_{SR_1}d_{R_1C}\right)^{-\alpha}y_p^{(3)}}{\varrho_C^2 P \left(d_{SR_1}d_{R_1C}\right)^{-\alpha}y_p^{(3)}+\sigma^2}-\tilde {\beta }}\right )}\right ) \nonumber\\
&\cdot \left ({\frac {1}{2}-k\left ({\frac{a_E P\left(d_{SR_1}d_{R_1C}\right)^{-\alpha}y_p^{(3)}}{(a_C+\varrho_C^2) P \left(d_{SR_1}d_{R_1C}\right)^{-\alpha}y_p^{(3)}+\sigma^2}-\beta }\right )}\right )\nonumber\\
&\cdot f_{|\gamma_{O}^{(C)}|^{2}}(y_p^{(3)}),
\end{align}
where $y_p^{(2)}=\frac{\phi _{v}-\phi _{\tilde {v}}}{2}t_p+\frac{\phi _{v}+\phi _{\tilde {v}}}{2}$ and $y_p^{(3)}=\frac{\phi _{\tilde {u}}-\phi _{v}}{2}t_p+\frac{\phi _{\tilde {u}}+\phi _{v}}{2}$.

{\it{Case 3:}} When $\phi _{\tilde {v}}\leq\phi_{v}\leq\phi_{u}<\phi_{\tilde {u}}$, we have
\begin{align}
&\mathbb {E}[\Psi \Phi |\gamma _{E,C}] \approx   F_{|\gamma_{O}^{(C)}|^{2}}(\phi _{\tilde {v}})+\frac{\phi _{v}-\phi _{\tilde {v}}}{2}\sum ^{U}_{p=1}\frac {\pi }{U}\sqrt {1-t_{p}^{2}} \nonumber\\
&\cdot \left ({\frac {1}{2}-\tilde {k}\left ({\frac{a_C P\left(d_{SR_1}d_{R_1C}\right)^{-\alpha}y_p^{(2)}}{\varrho_C^2 P \left(d_{SR_1}d_{R_1C}\right)^{-\alpha}y_p^{(2)}+\sigma^2}-\tilde {\beta }}\right )}\right )\nonumber\\
 & \cdot  f_{|\gamma_{O}^{(C)}|^{2}}(y_p^{(2)})+\frac{\phi _{\tilde {u}}-\phi _{v}}{2}\sum ^{U}_{p=1}\frac {\pi }{U}\sqrt {1-t_{p}^{2}} \nonumber\\
&\cdot \left ({\frac {1}{2}-\tilde {k}\left ({\frac{a_C P\left(d_{SR_1}d_{R_1C}\right)^{-\alpha}y_p^{(4)}}{\varrho_C^2 P \left(d_{SR_1}d_{R_1C}\right)^{-\alpha}y_p^{(4)}+\sigma^2}-\tilde {\beta }}\right )}\right ) \nonumber\\
&\cdot \left ({\frac {1}{2}-k\left ({\frac{a_E P\left(d_{SR_1}d_{R_1C}\right)^{-\alpha}y_p^{(4)}}{(a_C+\varrho_C^2) P \left(d_{SR_1}d_{R_1C}\right)^{-\alpha}y_p^{(4)}+\sigma^2}-\beta }\right )}\right )\nonumber\\
&\cdot f_{|\gamma_{O}^{(C)}|^{2}}(y_p^{(4)}),
\end{align}
where $y_p^{(4)}=\frac{1}{2k}t_p+\frac{\beta}{2}$.

{\it{Case 4:}} When $\phi _{\tilde {v}}\leq\phi_{v}<\phi_{\tilde {u}}\leq\phi_{u}$,   we have
\begin{align}
&\mathbb {E}[\Psi \Phi |\gamma _{E,C}] \approx   F_{|\gamma_{O}^{(C)}|^{2}}(\phi _{v})+\frac{\phi _{v}-\phi _{\tilde {v}}}{2}\sum ^{U}_{p=1}\frac {\pi }{U}\sqrt {1-t_{p}^{2}} \nonumber\\
&\cdot \left ({\frac {1}{2}-\tilde {k}\left ({\frac{a_C P\left(d_{SR_1}d_{R_1C}\right)^{-\alpha}y_p^{(5)}}{\varrho_C^2 P \left(d_{SR_1}d_{R_1C}\right)^{-\alpha}y_p^{(5)}+\sigma^2}-\tilde {\beta }}\right )}\right )\nonumber\\
 & \cdot  f_{|\gamma_{O}^{(C)}|^{2}}(y_p^{(5)})+\frac{\phi _{\tilde {u}}-\phi _{v}}{2}\sum ^{U}_{p=1}\frac {\pi }{U}\sqrt {1-t_{p}^{2}} \nonumber\\
&\cdot \left ({\frac {1}{2}-\tilde {k}\left ({\frac{a_C P\left(d_{SR_1}d_{R_1C}\right)^{-\alpha}y_p^{(1)}}{\varrho_C^2 P \left(d_{SR_1}d_{R_1C}\right)^{-\alpha}y_p^{(1)}+\sigma^2}-\tilde {\beta }}\right )}\right ) \nonumber\\
&\cdot \left ({\frac {1}{2}-k\left ({\frac{a_E P\left(d_{SR_1}d_{R_1C}\right)^{-\alpha}y_p^{(1)}}{(a_C+\varrho_C^2) P \left(d_{SR_1}d_{R_1C}\right)^{-\alpha}y_p^{(1)}+\sigma^2}-\beta }\right )}\right )\nonumber\\
&\cdot f_{|\gamma_{O}^{(C)}|^{2}}(y_p^{(1)}),
\end{align}
where $y_p^{(5)}=\frac{\phi _{\tilde {v}}-\phi _{v}}{2}t_p+\frac{\phi _{v}+\phi _{\tilde {v}}}{2}$.

{\it{Case 5:}} When $\phi_{v}\leq\phi _{\tilde {v}}\leq\phi_{u}<\phi_{\tilde {u}}$,  we have
\begin{align}
&\mathbb {E}[\Psi \Phi |\gamma _{E,C}] \approx   F_{|\gamma_{O}^{(C)}|^{2}}(\phi _{v})+\frac{\phi _{v}-\phi _{\tilde {v}}}{2}\sum ^{U}_{p=1}\frac {\pi }{U}\sqrt {1-t_{p}^{2}} \nonumber\\
&\cdot \left ({\frac {1}{2}-\tilde {k}\left ({\frac{a_C P\left(d_{SR_1}d_{R_1C}\right)^{-\alpha}y_p^{(5)}}{\varrho_C^2 P \left(d_{SR_1}d_{R_1C}\right)^{-\alpha}y_p^{(5)}+\sigma^2}-\tilde {\beta }}\right )}\right )\nonumber\\
 & \cdot  f_{|\gamma_{O}^{(C)}|^{2}}(y_p^{(5)})+\frac{\phi _{\tilde {u}}-\phi _{v}}{2}\sum ^{U}_{p=1}\frac {\pi }{U}\sqrt {1-t_{p}^{2}} \nonumber\\
&\cdot \left ({\frac {1}{2}-\tilde {k}\left ({\frac{a_C P\left(d_{SR_1}d_{R_1C}\right)^{-\alpha}y_p^{(6)}}{\varrho_C^2 P \left(d_{SR_1}d_{R_1C}\right)^{-\alpha}y_p^{(6)}+\sigma^2}-\tilde {\beta }}\right )}\right ) \nonumber\\
&\cdot \left ({\frac {1}{2}-k\left ({\frac{a_E P\left(d_{SR_1}d_{R_1C}\right)^{-\alpha}y_p^{(6)}}{(a_C+\varrho_C^2) P \left(d_{SR_1}d_{R_1C}\right)^{-\alpha}y_p^{(6)}+\sigma^2}-\beta }\right )}\right )\nonumber\\
&\cdot f_{|\gamma_{O}^{(C)}|^{2}}(y_p^{(6)}),
\end{align}
where $y_p^{(6)}=\frac{\phi _{u}-\phi _{\tilde {v}}}{2}t_p+\frac{\phi _{u}+\phi _{\tilde {v}}}{2}$.

{\it{Case 6:}} When $\phi_{v}\leq\phi_{u}\leq\phi _{\tilde {v}}\leq\phi_{\tilde {u}}$, utilizing the Gauss-Chebyshev integral, we have
\begin{align}
&\mathbb {E}[\Psi \Phi |\gamma _{E,C}]\approx \left(\frac {1}{2}+k {\beta }\right)\left(F_{|\gamma_{O}^{(C)}|^{2}}\left(\phi _{v}\right)-F_{|\gamma_{O}^{(C)}|^{2}}\left(\phi _{u}\right)\right)\nonumber\\
&+\sum ^{U}_{p=1}\frac {\pi }{U}\frac{\sqrt {1-t_{p}^{2}}}{2}{\frac{a_C P\left(d_{SR_1}d_{R_1C}\right)^{-\alpha}y_p^{(4)}}{\varrho_C^2 P \left(d_{SR_1}d_{R_1C}\right)^{-\alpha}y_p^{(4)}+\sigma^2}}f_{|\gamma_{O}^{(C)}|^{2}}(y_p^{(4)}).
\end{align}

The derivations for $\mathbb{E}[\Psi \Xi |\gamma_{E,C}]$ follow a similar structure with Cases 7-12, where the analysis proceeds analogously by considering the relative ordering of the corresponding critical values.

\section{Proof of Proposition 1}\label{Ap1}
From \cite{LaiX21}, the first order derivative of $\tilde{\beta}$ with respective to $\gamma_{E,C}$ is given by
\begin{align}
&\frac {\partial \tilde {\beta }}{\partial \gamma_{E,C}} =\exp \left ({\frac {\mu V^{\frac {1}{2}}(\gamma_{E,C})}{\sqrt {m}}\ln 2+\frac {N_{c}}{m}\ln 2}\right ) \\&\cdot \left [{1+\frac {\mu }{\sqrt {m}(1-(1+\gamma _{E,C})^{-2})^{\frac {1}{2}}(1+\gamma _{E,C})^{3}}}\right ].
\end{align}
Since $\mu=Q^{-1}(\delta)\geq 0$ when $\delta\in [0,12]$, we have
\begin{align}
\frac {\partial \tilde {\beta }}{\partial \gamma_{E,C}}\geq 0.
\end{align}
The expression $\tilde{\beta}$ is a monotonically increasing function with respect to $\gamma_{E,C}$. Furthermore, 
\begin{align}
	&\frac {\partial \gamma_{E,C}}{\partial |\gamma_{O}^{(E)}|^{2}} =\frac{1}{\left(\varrho_E^2 P \left(d_{SR_2}d_{R_2E}\right)^{-\alpha}\left|\gamma^{(E)}_O\right|^2+\sigma^2\right)^2}\geq 0.
\end{align}
Thus, $\gamma_{E,C}$ is a monotonically increasing function with respect to $|\gamma_{O}^{(E)}|^{2}$. Thus, $\tilde{\beta}$ is a monotonically increasing function with respect to $|\gamma_{O}^{(E)}|^{2}$.

\section{Detailed Derivations for Cases 7-12}\label{Ap3}
 
\subsection{Critical Points for Cases 7-12}

For Cases 7–12, the four critical points of $|\gamma_{O}^{(E)}|^{2}$ are the solutions to four equations as follows:
\begin{align} 
	\bar {\phi }_{\tilde {v}}=&\phi _{v},\label{aq85}\\ 
	\bar {\phi }_{\tilde {v}}=&\phi _{u}, \label{aq86}\\ 
	\bar {\phi }_{\tilde {u}}=&\phi _{v}, \label{aq87}\\ 
	\bar {\phi }_{\tilde {u}}=&\phi _{u}. \label{aq88}
\end{align}

Substituting the general form \eqref{aq30} into \eqref{aq85}, we have $\tilde {\beta }-1/(2\tilde {k})=\varsigma _{v}$, where $\varsigma _{v}=\frac {v a_C}{a_E(1+v)-a_C v}$.
Substituting $\tilde {k}=\sqrt{m}(2\pi \tilde {\beta}(\tilde {\beta}+2))^{-\frac{1}{2}}$ into this equation, after mathematical manipulations, we obtain:
\begin{align} \label{aq91}
	(2m-\pi )\tilde {\beta }^{2}-\left ({4m\varsigma _{v}+2\pi }\right )\tilde {\beta }+2m\varsigma _{v}^{2}=0.
\end{align}

Using Proposition 1, after obtaining two solutions of $\tilde {\beta }$, we obtain the solutions of $|\gamma_{O}^{(E)}|^{2}$ to \eqref{aq91}, denoted as $\tau_5$ and $\tau_6$ with $\tau_5\leq \tau_6$, by bisection search. Similarly, the solutions of $|\gamma_{O}^{(E)}|^{2}$ to \eqref{aq87} are also $\tau_5$ and $\tau_6$.

For equations \eqref{aq86} and \eqref{aq88}, we obtain:
\begin{align}\label{aq92}
	(2m-\pi )\tilde {\beta }^{2}-\left ({4m\varsigma _{u}+2\pi }\right )\tilde {\beta }+2m\varsigma _{u}^{2}=0,
\end{align}
where $\varsigma _{u}=\frac {u a_C}{a_E(1+u)-a_C u}$.
Using Proposition 1, the solutions of $|\gamma_{O}^{(E)}|^{2}$ to \eqref{aq92} are denoted as $\tau_7$ and $\tau_8$, with $\tau_7\leq \tau_8$.

\subsection{Detailed Case Derivations}

{\it{Case 7:}} When $\bar {\phi }_{\tilde {v}}<\bar{\phi }_{\tilde {u}}\leq\phi_{v}\leq\phi_{u}$, substituting the linear approximations into the conditional expectation, we have
\begin{align}
	&\mathbb {E}[\Psi \Xi |\gamma _{E,C}]\approx \left(\frac {1}{2}+\tilde {k}\tilde {\beta }\right)\left(F_{|\gamma_{O}^{(C)}|^{2}}\left(\bar {\phi }_{\tilde {v}}\right)+F_{|\gamma_{O}^{(C)}|^{2}}\left(\bar {\phi }_{\tilde {u}}\right)\right)\nonumber\\
	&+\sum ^{U}_{p=1}\frac {\pi }{U}\frac{\sqrt {1-t_{p}^{2}}}{2}{\frac{a_C P\left(d_{SR_1}d_{R_1C}\right)^{-\alpha}y_p^{(1)}}{\varrho_C^2 P \left(d_{SR_1}d_{R_1C}\right)^{-\alpha}y_p^{(1)}+\sigma^2}}f_{|\gamma_{O}^{(C)}|^{2}}(y_p^{(1)}),
\end{align}
where $y_p^{(1)}=\frac{1}{2\tilde {k}}t_p+\frac{\tilde {\beta }}{2}$.

{\it{Case 8:}} When $\bar {\phi }_{\tilde {v}}\leq\bar{\phi }_{\tilde {u}}<\phi_{v}\leq\phi_{u}$, we have
\begin{align}
	&\mathbb {E}[\Psi \Xi |\gamma _{E,C}] \approx   F_{|\gamma_{O}^{(C)}|^{2}}(\bar {\phi }_{\tilde {v}})+\frac{\phi _{v}-\bar {\phi }_{\tilde {v}}}{2}\sum ^{U}_{p=1}\frac {\pi }{U}\sqrt {1-t_{p}^{2}} \nonumber\\
	&\cdot \left ({\frac {1}{2}-\tilde {k}\left ({\frac{a_C P\left(d_{SR_1}d_{R_1C}\right)^{-\alpha}y_p^{(7)}}{\varrho_C^2 P \left(d_{SR_1}d_{R_1C}\right)^{-\alpha}y_p^{(7)}+\sigma^2}-\tilde {\beta }}\right )}\right )\nonumber\\
	& \cdot  f_{|\gamma_{O}^{(C)}|^{2}}(y_p^{(7)})+\frac{\bar {\phi }_{\tilde {u}}-\phi _{v}}{2}\sum ^{U}_{p=1}\frac {\pi }{U}\sqrt {1-t_{p}^{2}} \nonumber\\
	&\cdot \left ({\frac {1}{2}-\tilde {k}\left ({\frac{a_C P\left(d_{SR_1}d_{R_1C}\right)^{-\alpha}y_p^{(8)}}{\varrho_C^2 P \left(d_{SR_1}d_{R_1C}\right)^{-\alpha}y_p^{(8)}+\sigma^2}-\tilde {\beta }}\right )}\right ) \nonumber\\
	&\cdot \left ({\frac {1}{2}-k\left ({\frac{a_E P\left(d_{SR_1}d_{R_1C}\right)^{-\alpha}y_p^{(8)}}{(a_C+\varrho_C^2) P \left(d_{SR_1}d_{R_1C}\right)^{-\alpha}y_p^{(8)}+\sigma^2}-\beta }\right )}\right )\nonumber\\
	&\cdot f_{|\gamma_{O}^{(C)}|^{2}}(y_p^{(8)}),
\end{align}
where $y_p^{(7)}=\frac{\phi _{v}-\bar {\phi }_{\tilde {v}}}{2}t_p+\frac{\phi _{v}+\bar {\phi }_{\tilde {v}}}{2}$ and $y_p^{(8)}=\frac{\bar {\phi }_{\tilde {u}}-\phi _{v}}{2}t_p+\frac{\bar {\phi }_{\tilde {u}}+\phi _{v}}{2}$.

{\it{Case 9:}} When $\bar {\phi }_{\tilde {v}}\leq\phi_{v}\leq\phi_{u}<\bar{\phi }_{\tilde {u}}$, we have
\begin{align}
	&\mathbb {E}[\Psi \Xi |\gamma _{E,C}] \approx   F_{|\gamma_{O}^{(C)}|^{2}}(\bar {\phi }_{\tilde {v}})+\frac{\phi _{v}-\bar {\phi }_{\tilde {v}}}{2}\sum ^{U}_{p=1}\frac {\pi }{U}\sqrt {1-t_{p}^{2}} \nonumber\\
	&\cdot \left ({\frac {1}{2}-\tilde {k}\left ({\frac{a_C P\left(d_{SR_1}d_{R_1C}\right)^{-\alpha}y_p^{(7)}}{\varrho_C^2 P \left(d_{SR_1}d_{R_1C}\right)^{-\alpha}y_p^{(7)}+\sigma^2}-\tilde {\beta }}\right )}\right )\nonumber\\
	& \cdot  f_{|\gamma_{O}^{(C)}|^{2}}(y_p^{(7)})+\frac{\phi _{u}-\phi _{v}}{2}\sum ^{U}_{p=1}\frac {\pi }{U}\sqrt {1-t_{p}^{2}} \nonumber\\
	&\cdot \left ({\frac {1}{2}-\tilde {k}\left ({\frac{a_C P\left(d_{SR_1}d_{R_1C}\right)^{-\alpha}y_p^{(4)}}{\varrho_C^2 P \left(d_{SR_1}d_{R_1C}\right)^{-\alpha}y_p^{(4)}+\sigma^2}-\tilde {\beta }}\right )}\right ) \nonumber\\
	&\cdot \left ({\frac {1}{2}-k\left ({\frac{a_E P\left(d_{SR_1}d_{R_1C}\right)^{-\alpha}y_p^{(4)}}{(a_C+\varrho_C^2) P \left(d_{SR_1}d_{R_1C}\right)^{-\alpha}y_p^{(4)}+\sigma^2}-\beta }\right )}\right )\nonumber\\
	&\cdot f_{|\gamma_{O}^{(C)}|^{2}}(y_p^{(4)}),
\end{align}
where $y_p^{(4)}=\frac{1}{2k}t_p+\frac{\beta}{2}$.

{\it{Case 10:}} When $\bar {\phi }_{\tilde {v}}\leq\phi_{v}<\bar{\phi }_{\tilde {u}}\leq\phi_{u}$, we have
\begin{align}
	&\mathbb {E}[\Psi \Xi |\gamma _{E,C}] \approx   F_{|\gamma_{O}^{(C)}|^{2}}(\phi _{v})+\frac{\bar {\phi }_{\tilde {v}}-\phi _{v}}{2}\sum ^{U}_{p=1}\frac {\pi }{U}\sqrt {1-t_{p}^{2}} \nonumber\\
	&\cdot \left ({\frac {1}{2}-\tilde {k}\left ({\frac{a_C P\left(d_{SR_1}d_{R_1C}\right)^{-\alpha}y_p^{(9)}}{\varrho_C^2 P \left(d_{SR_1}d_{R_1C}\right)^{-\alpha}y_p^{(9)}+\sigma^2}-\tilde {\beta }}\right )}\right )\nonumber\\
	& \cdot  f_{|\gamma_{O}^{(C)}|^{2}}(y_p^{(9)})+\frac{\bar {\phi }_{\tilde {u}}-\bar {\phi }_{\tilde {v}}}{2}\sum ^{U}_{p=1}\frac {\pi }{U}\sqrt {1-t_{p}^{2}} \nonumber\\
	&\cdot \left ({\frac {1}{2}-\tilde {k}\left ({\frac{a_C P\left(d_{SR_1}d_{R_1C}\right)^{-\alpha}y_p^{(1)}}{\varrho_C^2 P \left(d_{SR_1}d_{R_1C}\right)^{-\alpha}y_p^{(1)}+\sigma^2}-\tilde {\beta }}\right )}\right ) \nonumber\\
	&\cdot \left ({\frac {1}{2}-k\left ({\frac{a_E P\left(d_{SR_1}d_{R_1C}\right)^{-\alpha}y_p^{(1)}}{(a_C+\varrho_C^2) P \left(d_{SR_1}d_{R_1C}\right)^{-\alpha}y_p^{(1)}+\sigma^2}-\beta }\right )}\right )\nonumber\\
	&\cdot f_{|\gamma_{O}^{(C)}|^{2}}(y_p^{(1)}),
\end{align}
where $y_p^{(9)}=\frac{\bar {\phi }_{\tilde {v}}-\phi _{v}}{2}t_p+\frac{\phi _{v}+\bar {\phi }_{\tilde {v}}}{2}$.

{\it{Case 11:}} When $\phi_{v}\leq\bar {\phi }_{\tilde {v}}\leq\phi_{u}<\bar{\phi }_{\tilde {u}}$, we have
\begin{align}
	&\mathbb {E}[\Psi \Xi |\gamma _{E,C}] \approx   F_{|\gamma_{O}^{(C)}|^{2}}(\phi _{v})+\frac{\bar {\phi }_{\tilde {v}}-\phi _{v}}{2}\sum ^{U}_{p=1}\frac {\pi }{U}\sqrt {1-t_{p}^{2}} \nonumber\\
	&\cdot \left ({\frac {1}{2}-\tilde {k}\left ({\frac{a_C P\left(d_{SR_1}d_{R_1C}\right)^{-\alpha}y_p^{(9)}}{\varrho_C^2 P \left(d_{SR_1}d_{R_1C}\right)^{-\alpha}y_p^{(9)}+\sigma^2}-\tilde {\beta }}\right )}\right )\nonumber\\
	& \cdot  f_{|\gamma_{O}^{(C)}|^{2}}(y_p^{(9)})+\frac{\phi _{u}-\bar {\phi }_{\tilde {v}}}{2}\sum ^{U}_{p=1}\frac {\pi }{U}\sqrt {1-t_{p}^{2}} \nonumber\\
	&\cdot \left ({\frac {1}{2}-\tilde {k}\left ({\frac{a_C P\left(d_{SR_1}d_{R_1C}\right)^{-\alpha}y_p^{(10)}}{\varrho_C^2 P \left(d_{SR_1}d_{R_1C}\right)^{-\alpha}y_p^{(10)}+\sigma^2}-\tilde {\beta }}\right )}\right ) \nonumber\\
	&\cdot \left ({\frac {1}{2}-k\left ({\frac{a_E P\left(d_{SR_1}d_{R_1C}\right)^{-\alpha}y_p^{(10)}}{(a_C+\varrho_C^2) P \left(d_{SR_1}d_{R_1C}\right)^{-\alpha}y_p^{(10)}+\sigma^2}-\beta }\right )}\right )\nonumber\\
	&\cdot f_{|\gamma_{O}^{(C)}|^{2}}(y_p^{(10)}),
\end{align}
where $y_p^{(10)}=\frac{\phi _{u}-\bar {\phi }_{\tilde {v}}}{2}t_p+\frac{\phi _{u}+\bar {\phi }_{\tilde {v}}}{2}$.

{\it{Case 12:}} When $\phi_{v}\leq\phi_{u}\leq\bar {\phi }_{\tilde {v}}\leq\bar{\phi }_{\tilde {u}}$, we have
\begin{align}
	&\mathbb {E}[\Psi \Xi |\gamma _{E,C}]\approx \left(\frac {1}{2}+k {\beta }\right)\left(F_{|\gamma_{O}^{(C)}|^{2}}\left(\phi _{v}\right)+F_{|\gamma_{O}^{(C)}|^{2}}\left(\phi _{u}\right)\right)\nonumber\\
	&+\sum ^{U}_{p=1}\frac {\pi }{U}\frac{\sqrt {1-t_{p}^{2}}}{2}{\frac{a_C P\left(d_{SR_1}d_{R_1C}\right)^{-\alpha}y_p^{(4)}}{\varrho_C^2 P \left(d_{SR_1}d_{R_1C}\right)^{-\alpha}y_p^{(4)}+\sigma^2}}f_{|\gamma_{O}^{(C)}|^{2}}(y_p^{(4)}).
\end{align}

\subsection{Integration Interval Analysis for Cases 7-12}

For Cases 7-12, the integral intervals in \eqref{aq70} are determined based on the relative ordering of the critical points $\tau_5, \tau_6, \tau_7, \tau_8$:

{\it{Case 7:}} When $\bar{\phi}_{\tilde {u}}\leq\phi_{v}$, the integral interval is $\left ({(0,\tau _{5}]\bigcup [\tau _{6},\infty )}\right )\bigcap [\tau _{0}, \bar {\tau }_{v}]$.

{\it{Case 8:}} The integral interval is $\Theta _{7}\bigcap\Theta _{8}\bigcap\Theta _{9}$, where
\begin{align}
	\Theta _{7}=&\left ({[\tau _{5}, \tau _{6}]\bigcap [\tau _{0}, \bar {\tau }_{v}]}\right )\bigcup [\bar {\tau }_{v}, \infty ),\\
	\Theta _{8}=&\left ({[\tau _{5}, \tau _{6}]\bigcap [\bar {\tau }_{v}, \infty )}\right )\bigcup [\tau _{0}, \bar {\tau }_{v}],\\ 
	\Theta _{9}=&\left ({(0,\tau _{7}]\bigcup [\tau _{8},\infty )}\right )\bigcap [\tau _{0}, \bar {\tau }_{u}].
\end{align}

{\it{Case 9:}} The integral interval is $\Theta _{7}\bigcap\Theta _{10}$, where $\Theta _{10}=\left ({[\tau _{7}, \tau _{8}]\bigcap [\tau _{0}, \bar {\tau }_{u}]}\right )\bigcup \big[\bar {\tau }_{u}, \infty \big )$.

{\it{Case 10:}} The integral interval is $\Theta _{9}\bigcap\Theta _{11}$, where $\Theta _{11}=\left ({(0,\tau _{5}]\bigcup [\tau _{6},\infty )}\right )\bigcap \big[\bar {\tau }_{v}, \infty \big)$.

{\it{Case 11:}} The integral interval is $\Theta _{10}\bigcap\Theta _{11}\bigcap\Theta _{12}$, where $\Theta _{12}=\left ({[\tau _{7}, \tau _{8}]\bigcap \big[\bar {\tau }_{u}, \infty \big )}\right )\bigcup [\tau _{0}, \bar {\tau }_{u}]$.

{\it{Case 12:}} The integral interval is $\left ({(0,\tau _{7}]\bigcup [\tau _{8},\infty )}\right )\bigcap \big[\bar {\tau }_{u}, \infty \big )$.

\end{document}